\documentclass[journal,transmag]{IEEEtran}
\usepackage[comma,numbers,square,sort&compress]{natbib}

\usepackage{amsmath}
\usepackage{amssymb}
\usepackage{float}
\usepackage{amsthm}
\usepackage{graphicx}
\usepackage{epstopdf}
\usepackage{color}
\usepackage{verbatim}
\usepackage{tikz,times}
\usepackage{algorithm}
\usepackage{algorithmic}
\newtheorem{definition}{\textbf{Definition}}
\newtheorem{remark}{\textbf{Remark}}
\newtheorem{theorem}{\textbf{Theorem}}
\newtheorem{lemma}{\textbf{Lemma}}
\newtheorem{assumption}{\textbf{Assumption}}
\newtheorem{corollary}{\textbf{Corollary}}
\newtheorem{proposition}{\textbf{Proposition}}
\allowdisplaybreaks[4]
\begin{document}

	\title{Distributed Kalman Filters with State Equality Constraints: Time-based and Event-triggered Communications}
	
	\author{Xingkang~He,
		~Chen~Hu,
		~Yiguang~Hong,
		\IEEEmembership{Fellow,~IEEE},~ Ling Shi, Haitao Fang
		
		\thanks{The work is supported by National Key Research and Development Program of China (2016YFB0901902), National Natural Science Foundation of China (61573344);
			The work by Ling Shi is supported by an RGC General Research Fund 16222716, and by an HKUST-KTH Partnership FP804.
	}
			\thanks{Xingkang He is with the  Division of Decision and Control Systems,  KTH Royal Institute of
				Technology, SE-100 44 Stockholm, Sweden. (xingkang@kth.se)}
		\thanks{Yiguang Hong and Haitao Fang  are with the Key Laboratory of Systems and Control,  Academy of Mathematics and Systems Science,
			Chinese Academy of Sciences, Beijing, China. (yghong,htfang@iss.ac.cn)}
		\thanks{Chen Hu is with Rocket Force University of Engineering, Shaanxi, China. (chenh628@hotmail.com)}
		\thanks{Ling Shi is with Department of Electronic and Computer Engineering, Hong Kong University of Science and Technology. (eesling@ust.hk)}
	}
	
	
	\maketitle

	\begin{abstract}
In this paper, we investigate a distributed estimation problem for
multi-agent systems with state equality constraints (SEC). First, under a time-based  consensus communication protocol, 
applying a modified projection operator and the covariance intersection fusion  method, 
we propose a distributed Kalman filter with guaranteed consistency and satisfied SEC. 
Furthermore, we  establish the relationship between  consensus step, SEC and estimation error covariance in dynamic  and steady processes, respectively.
Employing a space decomposition method,  we show the  error covariance in the constraint set can be  arbitrarily small by setting a sufficiently large  consensus step. 
 Besides, we propose an extended collective observability (ECO) condition based on SEC, which is milder than existing observability conditions. Under the  ECO condition, through utilizing a technique of matrix approximation, we   prove the boundedness of error covariance and the exponentially  asymptotic unbiasedness of state estimate, respectively.  
Moreover,  under  the ECO condition for linear time-invariant systems with SEC, we  provide a novel event-triggered communication protocol by employing the consistency, and give an offline design principle of triggering thresholds with guaranteed boundedness of error covariance.
More importantly, we quantify and analyze the communication rate for the proposed event-triggered distributed Kalman filter,  and provide optimization based methods to obtain the minimal (maximal) successive non-triggering (triggering) times. 
Two simulations are provided to demonstrate the developed theoretical results and  the effectiveness of the filters.
	\end{abstract}


	\section{Introduction}
	State estimation problems are very important, related to parameter identification, signal reconstruction, target monitoring and control design, which have been studied for several decades.
	In recent years, distributed state estimation of multi-agent systems has received more and more attention because of its broad range of applications in engineering systems such as communication networks, sensor networks and smart grids.
	Among the distributed estimation methods \cite{cattivelli2010diffusion, Zhang2016Distributed,Hu2012Diffusion,Zhu2016Distributed,olfati2009kalman,olfati2007distributed, yang2014stochastic, Battistelli2014Kullback,Battistelli2016stability,Battistelli2015Consensus,Dm2015Distributed,yang2017stochastic,carli2008distributed,ugrinovskii2011distributed}, Kalman-filter-based estimation plays a key role due to its ability of real-time estimation and non-stationary process tracking.
	
	Existing distributed Kalman filters can be roughly classified into two categories: the Kalman-consensus filter (KCF) and the Kalman-diffusion  filter (KDF).
	KCF is a design to include consensus terms to the Kalman filter structure (see \cite{olfati2009kalman,olfati2007distributed, yang2014stochastic, stankovic2009consensus, Battistelli2014Kullback,Battistelli2016stability,Battistelli2015Consensus,Dm2015Distributed,yang2017stochastic}).
	{For example, \cite{olfati2007distributed} constructed a KCF, with estimate of each agent obtained by consensus on measurement information to approximate centralized estimate. }
	Furthermore, \cite{stankovic2009consensus,yang2014stochastic} dealt with communication uncertainties like intermittent observation and communication faults.
	Nevertheless, \cite{olfati2009kalman,yang2014stochastic} assumed the local observability (i.e., the dynamical system is observable for each agent), which may be unsuitable in large-scale networks. Even though \cite{olfati2007distributed} took a collective observability condition, the agents have to communicate infinity times between any two sampling instants.   Some works extended the local observability assumption to a collective observability assumption for the filter design.   For example,
	\cite{Dm2015Distributed} constructed a \textit{consensus+innovations} distributed estimator based on some knowledge of a global observation model, under a collective observability condition.
	\cite{Battistelli2014Kullback} considered the KCF with consensus on the inverse covariance matrix and the information vector, while \cite{Battistelli2016stability} extended the results to a class of nonlinear systems.
	Under collective observability assumptions, \cite{Battistelli2014Kullback,Battistelli2016stability} proved the upper boundedness of the error covariance matrix of the KCFs, with requirement of  the non-singularity of the system matrix.
	Additionally, \cite{zhou2012distributed,zhou2013distributed} studied the design for switching communication topologies.
	{In \cite{kar2011gossip}, the authors studied gossip interactive Kalman filter for networked systems based on a random communication scheme, and provided some theoretical results on the asymptotic properties of error process.  }
	On the other hand, the KDF is to fuse the neighbor's estimates updated by standard Kalman filter (referring to \cite{cattivelli2010diffusion, Hu2012Diffusion,Wenling2015Diffusion,xiao2004fast}).
	Compared with many results on KCF \cite{olfati2009kalman, yang2014stochastic, yang2017stochastic,wang2017convergence}, the designs of KDF usually do not require the computation of correlation matrices between agents and  do not need that all agents should obtain the same estimate in steady-state.
	For example, \cite{cattivelli2010diffusion} discussed the design of the fusion weights and the convergence of the KDF algorithm, while \cite{Hu2012Diffusion} proposed a KDF based on a covariance intersection (CI) scheme.
	However, the observability assumptions of \cite{cattivelli2010diffusion,Hu2012Diffusion,Wenling2015Diffusion} are not global conditions.  
	{An effective distributed filter based on optimized weights for time-varying topology was proposed in \cite{wang2017convergence}.}	
	Although many effective distributed Kalman filters have been developed,  the stability research of distributed filters relying mild collective observability conditions still need more attention.
	
	In practice, we may obtain some information or knowledge on the state constraints in advance from physical laws, geometric relationship, and environment constraints of the considered system, which can improve the estimation performance in the filter design. {Many practical examples can be formulated with state equality constraints (SEC), such as quternion-based attitude estimation \cite{crassidis2003unscented}, magnetohydrodynamic data assimilation \cite{Chandrasekar2004State}, tracking and navigation \cite{wen1992model,shen2006reliable} and aeronautics \cite{simon2002kalman,rotea2008state}.
		A typical example of linear equality constraints arises in target tracking problem.
		A car may be traveling off-road, or on an unknown road, where the problem is unconstrained. However, most of the time it may be traveling along a given road, where the estimation problem is constrained. Another example is tracking a train, where the railway can be treated as constraints.} There are various methods incorporating the information of state constraints into the Kalman filter structure, such as the pseudo-observation \cite{julier2007kalman}, projection \cite{simon2002kalman,Te2009state, ko2007state}, and moving horizon \cite{rao2001constrained}. A survey on the conventional design of the Kalman filter with state constraints  provided in \cite{simon2010kalman} showed the usefulness of the constraints as  additional information to improve the estimation performance or accuracy.  It is known that the moving horizon method performs the best, but at the cost of high computational complexity  \cite{rao2001constrained,simon2010kalman}. {In some cases, one may formulate state constraints as equalities and project the unconstrained estimate onto the constrained surface under SEC \cite{simon2010kalman,Te2009state}.
		Linear state inequality constraints can be transformed to SEC in some conditions by using active set methods \cite{simon2010kalman}.	
		}
	Although the Kalman filter with state constraints has drawn much research attention, the corresponding distributed version,  i.e., distributed Kalman filter with state constraints, has not yet been adequately investigated, to the best of our knowledge.
	
	Communication may spend much cost or energy in distributed design, and therefore, a well-known scheme, called event-triggered scheme, is developed to reduce the communication costs, where the communication is carried out only when some predefined event conditions are satisfied.
	{A variance-based triggering  scheme of state estimation was studied in \cite{trimpe2014event}, which analyzed the convergence of the switching Riccati equation of estimation variance   for a scalar system.}
	With the Gaussian properties of the \textit{a priori} conditional distribution, a deterministic event-triggered schedule was proposed in \cite{wu2013event}, and moreover, to overcome the limitation of the Gaussian assumption, a stochastic event-triggered scheme was developed and the corresponding exact minimum mean square error estimator was obtained in \cite{HanTacStochasEV}.  Additionally, a set-valued measurement Kalman filter  for remote state estimation problem was studied in   \cite{shi2015set}, which built some relationship between the estimation performance and triggering thresholds.
	Also, Send-on-Delta (SoD), a typical event-based regulation proposed in \cite{miskowicz2006send}, has been widely utilized in the triggering condition design of estimation algorithms \cite{liu2015event,andren2016event,nguyen2007improving}.
	A recursive distributed filter with SoD triggering scheme on innovation was studied in \cite{liu2015event}, where the relationship between triggering thresholds and an upper bound of error covariance was analyzed. Additionally, a stochastic SoD triggering rule was proposed in \cite{andren2016event}, where an upper bound and a lower bound of the error covariance were derived, but the performance was not studied.
	Although quite a few effective event-triggered filters were developed, the  relationship between the estimation performance and the triggering thresholds still need further investigations in a distributed framework under collective observable condition.
To the best of our knowledge, the the quantification and analysis of communication rate for distributed event-triggered filtering have not been adequately investigated, and very few studies on analyzing the communication rate, though many distributed event-triggered filtering have been developed.

	The objective of this paper is to provide the theoretical analysis on  distributed estimation  with SEC for filters based on time-based and event-triggered communications, respectively.
	The main contributions of this paper are fourfold:
	\begin{enumerate}
		\item 
		We propose a time-based distributed Kalman filter based on a consensus communication protocol for a class of stochastic systems with SEC. 
		Moreover, we prove  both  estimation consistency and local SEC  hold for each communication by employing the CI method and a modified projection operator. 
		Besides, 
		the filter does not require the spatial independence between  measurement noises of agents, which is required by many existing results  \cite{Dm2015Distributed,Cat2010Diffusion,wang2017convergence,olfati2009kalman}.

		\item 		We establish the relationship between  consensus step, SEC and  error covariance in dynamic process and steady process, respectively.
			Besides, 	employing a space decomposition method,  we show the error covariance in the constraint set can be  arbitrarily small by setting a sufficiently large  consensus step.  
			To the best knowledge of the authors, these are the first results to analyze the influence of the multi-step communications to the estimation performance of distributed filters \cite{Battistelli2016stability,Battistelli2014Kullback,Dm2015Distributed,Cat2010Diffusion}.

			\item  We show that the SEC can relax the mildest collective observability assumptions given in \cite{Battistelli2014Kullback,Battistelli2015Consensus,Battistelli2016stability,Dm2015Distributed,wang2017convergence} or the local observability conditions given in \cite{olfati2009kalman,stankovic2009consensus,yang2014stochastic}, and lead to the extended collective observability (ECO) condition. Under the  ECO condition, through utilizing a technique of matrix approximation,		
			we  prove  the boundedness of error covariance and the exponentially  asymptotic unbiasedness of state estimate, respectively. 
			Thus, the essential filtering properties have been extended for a larger class of stochastic time-varying systems.

		\item  Under  the ECO condition for linear time-invariant systems with SEC, we  provide a novel event-triggered communication protocol by employing the consistency, and give an offline design principle of triggering thresholds with guaranteed boundedness of error covariance.
			More importantly,  we quantify and analyze the communication rate for the proposed event-triggered distributed Kalman filter,  and provide optimization based methods to obtain the minimal (maximal) successive non-triggering (triggering) times, which seem to be not theoretically  investigated in the existing results \cite{liu2015event,liu2015eventinformatics,Battistelli2016Distributed} of distributed filtering to our knowledge.

	\end{enumerate}

	The remainder of the paper is organized as follows. The problem formulation is given in section II.
	Then a fully distributed Kalman filter for a time-varying system with SEC is proposed and analyzed in section III, while
	a distributed Kalman filter for a time-invariant system with a novel event-triggered communication  scheme is given and investigated in section IV.   Following that,
	numerical simulations on a constrained moving vehicle are shown in section V. Finally, the conclusions of this paper are provided in section VI.
	
	Notations: The superscript ``T" represents the transpose.  If $A$ and $B$ are both symmetric matrices, $A\geq B$ (or $A>B$) means that $A-B$ is a positive semidefinite (or positive definite) matrix. $I_{n}$ stands for the identity matrix with $n$ rows and $n$ columns, and $\otimes$ stands for the  Kronecker product. $E\{x\}$ denotes the mathematical expectation of the stochastic variable $x$. $blockdiag\{\cdot\}$ and $diag\{\cdot\}$ represent the diagonalizations of block elements and scalar elements, respectively. $tr(P)$ is the trace of the matrix $P$ and $var(x)$ is the variance of $x$.
	The integer set from $a$ to $b-1$ is denoted as $[a:b)$, and the integer set from $a$ to $b$ is $[a:b]=[a:b)\cup\{b\}$. $G^-$ is the Moore-Penrose inverse of $G$, and if $G$ is a nonsingular matrix, $G^-=G^{-1}$.
	{$dis(x,y)$ stands for the Euclidean distance between $x$ and $y$. $P_{D}[x]$ means the projected value of $x$ onto the set $D$.
$\mathbb{R}$ and $\mathbb{Z}^+$ stand for the set of real scalars and positive integers, respectively. $\mathbb{N}(\mu,\sigma^2)$ represents the probability distribution with the mean $\mu$ and variance $\sigma^2$. 1.0e4 stands for $1.0\times 10^4$. 	}
	For an operator $h(x)$, suppose $h^2(x)=h(h(x))$. $\lambda_{max}(A)$ is the maximal eigenvalue of matrix $A$. $\lceil \cdot\rceil$ is the ceiling function, and $\lfloor\cdot\rfloor$ is the floor function.
	
	\section{Problem Formulation and  Filter Design}
	
	In this section, we provide some preliminary knowledge and formulate a distributed Kalman filter problem with SEC.   Then we provide a filtering structure for the design.
	
	\subsection{Problem Statement}
	Consider the following time-varying stochastic dynamics
	\begin{align}\label{system_all}
	x_{k+1}=A_{k}x_{k}+\omega_{k},
	\end{align}
	where $x_{k}\in\mathbb{R}^{n}$ is the state vector,
	$A_{k}\in\mathbb{R}^{n\times n}$ is the known system matrix, subject to $\beta_{2}I_n\leq A_{k}A_{k}^T\leq\beta_{1}I_n$, $\beta_1,\beta_2>0$. $\omega_{k}\in\mathbb{R}^{n}$ is the zero-mean process noise, and $E\{\omega_{k}\omega_{k}^T\}\leq Q_{k}\in\mathbb{R}^{n\times n}$, where $0<\underline Q\leq Q_{k}\leq \overline Q<+\infty$. The state $x_{k}$ is observed by a multi-agent network of $N$ agents, whose measurement model is given as
	\begin{align}\label{system_all1}
	y_{k,i}=H_{k,i}x_{k}+v_{k,i},i=1,2,\ldots,N,
	\end{align}
	where $y_{k,i}\in\mathbb{R}^{m_i}$ is the measurement vector obtained by agent $i$,  $H_{k,i}\in\mathbb{R}^{m_i\times n}$ is the observation matrix of agent $i$ and $v_{k,i}\in\mathbb{R}^{m_i}$ is the zero-mean observation noise.
	{Suppose that $x_{0}$, $\{\omega_{k}\}_{k=0}^{\infty}$ and $\{v_{k,i}\}_{k=0}^{\infty}$ are Gaussian and temporally independent. 
	Also, $E\{(x_{0}-E\{x_{0}\})(x_{0}-E\{x_{0}\})^T\}\leq P_{0}\in\mathbb{R}^{n\times n}$ and $E\{v_{k,i}v_{k,i}^T\}\leq R_{k,i}\in\mathbb{R}^{m_i\times m_i}$,  where  $R_{k,i}>0$.
	Different from \cite{wang2017convergence,cattivelli2010diffusion,olfati2009kalman,Dm2015Distributed} that the measurement noises of agents are spatially independent, in this work, we study a more general case that  $v_{k,i}$ and $v_{k,j}$ may be correlated  with unknown correlation matrix at each time  $k$.
}

	The communication between agents in the multi-agent network is modeled as a directed graph $\mathcal{G=(V,E,A)}$, which consists of  the set of agents or nodes $\mathcal{V}=\{1,2,\ldots,N\}$, the set of edges $\mathcal{E}\subseteq \mathcal{V}\times \mathcal{V}$, and the weighted adjacent matrix $\mathcal{A}=[a_{i,j}]$.  In the weighted adjacent matrix $\mathcal{A}$, all the elements are nonnegative, row stochastic and the diagonal elements are all positive, i.e., $a_{i,i}> 0,a_{i,j}\geq 0,\sum_{j\in \mathcal{V}}a_{i,j}=1$. If $a_{i,j}>0,j\neq i$, there is an edge $(i,j)\in \mathcal{E}$, which means node $i$ can directly receive the information of node $j$, and node $j$ is called the in-neighbor of node $i$. Meanwhile, node $i$ is called the out-neighbor of node $j$. 
	$\mathcal{N}_{i,out}^0$ is set of out-neighbors of agent $i$ without including itself. In this paper, we call in-neighbor as neighbor without specific statement.	
	All the neighbors of node $i$ can be represented by the set $\{j\in\mathcal{V}|(i,j)\in \mathcal{E}\}\triangleq\mathcal{N}_{i}^0$, whose size  is denoted as $|\mathcal{N}_{i}^0|$. 
	Also, $\mathcal{N}_{i}^0\bigcup\{i\}\triangleq \mathcal{N}_{i}$. $\mathcal{G}$ is called strongly connected if for any pair nodes $(i_{1},i_{l})$, there exists a directed path from $i_{1}$ to $i_{l}$ consisting of edges $(i_{1},i_{2}),(i_{2},i_{3}),\ldots,(i_{l-1},i_{l})$. An undirected graph $\mathcal{G}$ is simply called connected if it is strongly connected.

		{In practical applications, many examples can be formulated with SEC \cite{crassidis2003unscented,Chandrasekar2004State,wen1992model,shen2006reliable,simon2002kalman,rotea2008state}.
	Given the dynamics (\ref{system_all}), there may exist equality constraints on the overall state for each agent. The SEC can be expressed as}	
	\begin{equation}\label{system_constriants}
	D_{k,i}x_{k}=d_{k,i},i=1,2,\ldots,N,
	\end{equation}
	where $D_{k,i}\in\mathbb{R}^{s_{k,i}\times n}$, $d_{k,i}\in\mathbb{R}^{s_{k,i}}$ are known constraint matrix and constraint vector of agent $i$, respectively, and $s_{k,i}$ is the number of constraints of agent $i$ at time $k$.
	Clearly, the number of constraints is no larger than  state dimension, i.e., $s_{k,i}\leq n.$ Also, $\mathcal{D}_{k}=\bigcap_{i\in\mathcal{V}} \mathcal{D}_{k,i}\neq \emptyset$, where $\mathcal{D}_{k,i}=\{x_k|D_{k,i}x_k=d_{k,i}\}$. The constraint equality $\bar D_kx_k=d_k$ contains all the local constraints $D_{k,i}x_k= d_{k,i},i\in\mathcal{V}$.
	{
	Without loss of generality, $D_{k,i}$ is supposed to be either  zero  or of full row  rank, and $\bar D_k\in\mathbb{R}^{\bar s_k\times n}$ is assumed to be of full row rank, subject to $0<\varpi I_{ \bar s_k}\leq \bar D_{k}\bar D_{k}^T\leq I_{\bar s_k}$. }

	For convenience, we denote
	\begin{align*}
	&\Phi_{k,k}=I_{n},\Phi_{k+1,k}=A_{k},\Phi_{j,k}=\Phi_{j,j-1}\cdots \Phi_{k+1,k}(j>k),\\
	&H_{k}=[H_{k,1}^T,H_{k,2}^T,\ldots,H_{k,N}^T]^T,\\
	&D_{k}=[D_{k,1}^T,D_{k,2}^T,\ldots,D_{k,N}^T]^T,\\
	&R_{k}=blockdiag\{R_{k,1},R_{k,2},\ldots,R_{k,N}\}.
	\end{align*}
	
		In the distributed framework, each agent aims to timely estimate the dynamics (\ref{system_all}) based on the knowledge of local measurements (\ref{system_all1}) and constraints (\ref{system_constriants}) as well as the information communication with neighbors. To solve the problem, the following assumptions are adopted throughout the paper.
	
	\begin{assumption}\label{ass_topology}
		The directed graph $\mathcal{G=(V,E,A)}$ of multi-agent network is strongly connected.
	\end{assumption}
	
	Assumption \ref{ass_topology} is a basic condition for the directed topology graph of multi-agent network (referring to \cite{yang2017stochastic}), which is important for the distributed design. If the topology graph $\mathcal{G=(V,E,A)}$ is undirected, the assumption degenerates to the connectedness of the network discussed in \cite{Battistelli2016Distributed,Dm2015Distributed}.
	
	\begin{assumption}\label{ass_observable}(Extended collective observability, ECO)
		There exist a positive integer $\bar N$ and a positive constant $\alpha$ such that for any $k\geq 0$,
		\begin{equation}\label{Observability_matrix}
		\sum_{j=k}^{k+\bar N}\Phi^T_{j,k}(H_{j}^TR_{j}^{-1}H_{j}+D_{j}^TD_{j})\Phi_{j,k}\geq \alpha I_n>0.
		\end{equation}	
	\end{assumption}	
	Even if the observability conditions in {\cite{Battistelli2014Kullback,Battistelli2016stability,Cat2010Diffusion,yang2017stochastic,Hu2012Diffusion,wang2017convergence} are not satisfied}, the ECO condition in Assumption $\ref{ass_observable}$  can still be fulfilled.
	In other words, the ECO condition is more general than the existing observability conditions \cite{Battistelli2014Kullback,Battistelli2015Consensus,Battistelli2016stability,Dm2015Distributed,wang2017convergence,olfati2009kalman,stankovic2009consensus,yang2014stochastic}
	 because the mildest observability conditions given in the existing papers such as \cite{Battistelli2014Kullback,Battistelli2016stability} still depend on $(A_{k},H_{k})$, which is a special case of Assumption $\ref{ass_observable}$ satisfied if $D_{k,i}=0,i\in\mathcal{V},k=1,2,\ldots$.

	\subsection{Filtering Structure}
	
	Given the dynamics (\ref{system_all}) with observations of $N$ agents (\ref{system_all1}) and constraints (\ref{system_constriants}), we propose the following general distributed filtering structure in Table \ref{filter_stru}.

	\begin{table}[htbp]
		\caption{Distributed filtering structure}
		\label{filter_stru}	
		\centering  
		\begin{tabular}{l}  
			\hline\hline
			\textbf{Prediction:}\\
			$\bar x_{k,i}=A_{k-1}\hat x_{k-1,i},$\\
			\textbf{Measurement update:}\\	
			$\tilde x_{k,i}=\bar x_{k,i}+K_{k,i}(y_{k,i}-H_{k,i}\bar x_{k,i})$,  \\
			{\textbf{Perform $L$ steps of fusion-projection (i.e, consensus):}}\\
			{$\quad$\textbf{Local fusion:} }\\
			{$\quad$$\check  x_{k,i}^l=\sum_{j\in \mathcal{N}_{i}}W_{k,i,j}^l\tilde x_{k,j}^{l,t},$ $l=0,1,\cdots,L-1$}\\  
			$\quad$\textbf{Projection:}\\
			{$\quad$$\tilde x_{k,i}^{l+1}=\arg\min_{x}(x-\check  x_{k,i}^l)^T(M_{k,i}^l)^{-1}(x-\check  x_{k,i}^l), \text{s.t.} D_{k,i}x=d_{k,i}.$}\\        
	\textbf{Output:}\\	
	$\hat x_{k,i}=\tilde x_{k,i}^L$.	\\		
			\hline
		\end{tabular}
	\end{table}
	{In Table \ref{filter_stru}, $\bar x_{k,i}$, $\tilde x_{k,i}$, $\check x_{k,i}^l$,  $\tilde x_{k,i}^l$   and $\hat x_{k,i}$ are the state estimates in the state prediction,  state update,  state fusion, state projection and state output  of agent $i$ at time $k$, respectively.
		The corresponding estimation errors are $\bar e_{k,i}$, $\tilde e_{k,i}$, $\check e_{k,i}^l$,  $\tilde e_{k,i}^l$   and $e_{k,i}$.  Additionally, $\tilde x_{k,j}^{l,t},j\in\mathcal{N}_i$, are the estimates to be fused by agent $i$ under time-based or event-triggered schemes. The matrices $K_{k,i}$, $W_{k,i,j}^l$ and $M_{k,i}^l$ are to be designed.}

	\begin{remark}
		The projection scheme in Table \ref{filter_stru}  has the advantage that the state estimates of one agent can be
		strictly satisfied with the SEC of its own (\ref{system_constriants}). 
	\end{remark}

		It is not trivial to study how to combine the projection method and the distributed filters without SEC, since to obtain an estimate with smaller covariance, the typical projection operation requires the error covariance matrix. 	However, in the general distributed Kalman filters \cite{olfati2009kalman,olfati2007distributed, yang2014stochastic, Dm2015Distributed, cattivelli2010diffusion, Hu2012Diffusion}, due to the unknown correlation between estimates of state,  the error covariance can not be obtained in a distributed manner.  
	Thus, we introduce the definition of consistency  as follows.
	\begin{definition}\cite{Julier1997A}\label{def_consis}
		(Consistency) Suppose $x_{k}$ is a random vector and $\hat x_{k}$ is the estimate of $x_{k}$. Then the pair ($\hat x_{k},P_{k}$) is said to be consistent  if $E\{(\hat x_{k}-x_{k})(\hat x_{k}- x_{k})^T\}\leq P_{k}.$
	\end{definition}
	 
	Consider the system (\ref{system_all})-(\ref{system_constriants}) with distributed filtering structure in Table \ref{filter_stru}.
	 In the subsequent, we try to answer the following important problems:
	\begin{enumerate}
		\item How to design $K_{k,i}$, $W_{k,i,j}^l$ and $M_{k,i}^l$ of Table \ref{filter_stru} such that the proposed distributed Kalman filter is fully distributed  and subject to consistency at each moment and at each fusion-projection step?
		\item Under what mild conditions the uniform boundedness of error covariance can be guaranteed?		
		What benefits can SEC  provide? More importantly, how does the fusion-projection number $L$ influence the estimation performance?
		\item How to provide a design principle for the event-triggered thresholds so as to guarantee the desired estimation performance? How to quantify and analyze the communication rate of a distributed filter with even-triggered communications?
	\end{enumerate}
	
	\section{Distributed Filter with Time-based Communications}
	In this section, we will focus on solving the problems in 1) and 2) of last section. Specifically, we will provide a distributed design of filter, analyze the main estimation properties, and study the benefits of SEC and fusion-projection number $L$ to the estimation performance.
	\subsection{Time-based Distributed Filter}
	Regarding the distributed filtering structure in Table \ref{filter_stru}, 
	suppose the time-based  scheme allows each agent communicates with its neighbors for $L\geq 1$ times  between two measurement updates. In Table \ref{ODKF2}, 
	applying a modified projection operator and the covariance intersection fusion  method, 	
	we propose a distributed algorithm called ``time-based projected distributed Kalman filter" (TPDKF). 
	\begin{table}[htbp]
		\caption{Time-based Projected Distributed Kalman Filter (TPDKF)}
		\label{ODKF2}	
		\centering  
		\begin{tabular}{l}  
			\hline\hline
			\textbf{Initialization:}\\
			{
			$(\hat x_{0,i},P_{0,i})$ is satisfied with $P_{0,i}\geq (1+\theta_i)P_0+\frac{\theta_i+1}{\theta_i}P_{0,i}^*$,$\exists \theta_i>0, $}\\
			{where $P_{0,i}^*\geq (\hat x_{0,i}-E\{x_0\})(\hat x_{0,i}-E\{x_0\})^T$}\\
			\textbf{Input:}\\	
			$(\hat x_{k-1,i},P_{k-1,i},\varepsilon_i)$,  \\
			\textbf{Prediction:}\\
			$\bar x_{k,i}=A_{k-1}\hat x_{k-1,i},$\\  
			$\bar P_{k,i}=A_{k-1}P_{k-1,i}A_{k-1}^T+Q_{k-1},$\\         
			\textbf{Measurement update:}\\
			$\tilde x_{k,i}=\bar x_{k,i}+K_{k,i}(y_{k,i}-H_{k,i}\bar x_{k,i}),$\\        
			$K_{k,i}=\bar P_{k,i}H_{k,i}^T(H_{k,i}\bar P_{k,i}H_{k,i}^T+R_{k,i})^{-1},$\\
			$\tilde P_{k,i}=(I-K_{k,i}H_{k,i})\bar P_{k,i}$,\\
			{\textbf{Perform $L$ steps of fusion-projection with $\tilde x_{k,i}^0=\tilde x_{k,i},\tilde P_{k,i}^0=\tilde P_{k,i}$:}}\\
			{$\qquad$\textbf{Local fusion:} Receive ($\tilde x_{k,j}^l$, $\tilde P_{k,j}^l$ ) from neighbors $j\in \mathcal{N}_{i}$}\\
			{$\qquad$$\check x_{k,i}^l=\check P_{k,i}^l\sum_{j\in \mathcal{N}_{i}}a_{i,j}(\tilde P_{k,j}^l)^{-1}\tilde x_{k,j}^l$, $l=0,1,\cdots,L-1,$}\\	
			{$\qquad$$\check P_{k,i}^l=(\sum_{j\in \mathcal{N}_{i}}a_{i,j}(\tilde P_{k,j}^l)^{-1})^{-1}$, $l=0,1,\cdots,L-1,$}	\\	
			{$\qquad$\textbf{Projection:}}\\
			{$\qquad$$\tilde x_{k,i}^{l+1}=\check x_{k,i}^l-\check P_{k,i}^lD_{k,i}^T(D_{k,i}\check P_{k,i}^lD_{k,i}^T)^{-}(D_{k,i}\check x_{k,i}^l-d_{k,i}),$}\\
			{$\qquad$$\tilde P_{k,i}^{l+1}=\check P_{k,i}^l-\check P_{k,i}^lD_{k,i}^T(D_{k,i}\check P_{k,i}^lD_{k,i}^T+\varepsilon_i I_{s_{k,i}})^{-1}D_{k,i}\check P_{k,i}^l$,}\\
			\textbf{Output:}\\	
			{$\hat x_{k,i}=\tilde x_{k,i}^L,P_{k,i}=\tilde P_{k,i}^L$.}	\\
			\hline
		\end{tabular}
	\end{table}
On 	TPDKF, we have the following lemma.
	\begin{lemma}\label{lem_singular}
		For the proposed TPDKF,   $\tilde P_{k,i}^{l+1}$ is a positive definite matrix, subject to $\hat P_{k,i}^{l+1}\leq\tilde P_{k,i}^{l+1}\leq\check P_{k,i}^{l}$,
		where $\hat P_{k,i}^{l+1}=\check P_{k,i}^l-\check P_{k,i}^lD_{k,i}^T(D_{k,i}\check P_{k,i}^lD_{k,i}^T)^{-}D_{k,i}\check P_{k,i}^l.$
Furthermore, if $D_{k,i}$ is of rank $d_{k,i}\leq n$, $\hat P_{k,i}^{l+1}$ has $d_{k,i}$ eigenvalues of zero.
	\end{lemma}
\begin{proof}
	The proof is given in Appendix \ref{pf_lem_singular}.
\end{proof}
	The positive parameter $\varepsilon_i$ in TPDKF is adjustable, and it can be set sufficiently small like the setting in \cite{Te2009state}. Besides, the employment of $\varepsilon_i$ in the modified projection operator ensures that $\tilde P_{k,i}^{l+1}$ is a positive definite matrix, which suits to the requirement of the CI based fusion operator.

		{	
	\begin{remark}
The matrix $P_{k,i}$  in TPDKF does not stand for the error covariance of agent $i$. The relationship between $P_{k,i}$ and the error covariance matrix will be shown in Lemma \ref{thm_pro}.
		\end{remark}	
		    \begin{remark}
                 If the fusion-projection number $L$ is sufficiently large, $\hat x_{k,i},i\in\mathcal{V},$ will go to consensus, which means global SEC will be satisfied by each agent. However, since it is not practical to make agents communicate for too many times between two measurement updates in dynamic estimation, it is essential to analyze the  influence of the fusion-projection to the estimation performance, which will be analyzed in Theorem \ref{prop_const} and Theorem \ref{thm_consistent}. 
		    \end{remark}		
	\subsection{Estimation Performance of TPDKF}		
	In this subsection, we will investigate the main estimation performance of TPDKF in several aspects.
	The following result shows the Gaussianity of the estimation error of TPDKF.
	
	\begin{proposition}\label{thm_distri}
		{
		For TPDKF of the system (\ref{system_all})--(\ref{system_all1}) with constraints (\ref{system_constriants}), the estimation error $e_{k,i}=\hat x_{k,i}-x_{k}$ is Gaussian, for any $i\in \mathcal{V}, k=1,2,\ldots$.}
	\end{proposition}
	
	\begin{proof}
		See the proof in Appendix \ref{pf_pro1}.
	\end{proof}

	Due to the strong correlation between the estimates of agents, error covariance matrices are usually not accessible.   The following lemma shows that both   consistency  and SEC are satisfied for the estimates of TPDKF. 
	\begin{lemma}\label{thm_pro}
			 For TPDKF of the system (\ref{system_all})--(\ref{system_all1}) with constraints (\ref{system_constriants}), if $L\geq 1$, then for $l\in[0:L-1]$,
		\begin{itemize}
			\item the state estimate $\tilde x_{k,i}^{l+1}$  satisfies the SEC of agent $i$ at time $k$, i.e., $D_{k,i}\tilde x_{k,i}^{l+1}=d_{k,i}$;
			\item the consistency at the modified projection step holds, i.e., $E\{(\tilde x_{k,i}^{l+1}-x_{k})(\tilde x_{k,i}^{l+1}-x_{k})^T\}\leq \tilde P_{k,i}^{l+1}$.
		\end{itemize}		
	\end{lemma}	
	\begin{proof}
		See the proof in Appendix \ref{pf_th1}.
	\end{proof}	
To investigate the influence of SEC and fusion-projection number $L$ to the estimation performance of TPDKF, we utilize a space decomposition method in the following.
For each time $k$,
there exists a bounded non-singular matrix $F_{k}\in\mathbb{R}^{n\times n}$, such that $\bar D_{k}F_{k}=[0^{\bar s_k\times (n-\bar s_k)}, \tilde D_{k}]$, where 
$\tilde D_{k}\in\mathbb{R}^{\bar s_k\times \bar s_k}$ is of full rank.  
Let 
\begin{align*}
&F_{k}^{-1}x_{k}=[\underline x_{k}^T(1),\underline x_{k}^T(2)]^T\\
&F_{k}^{-1}\hat x_{k,i}=[\underline {\hat x}_{k,i}^T(1),\underline {\hat x}_{k,i}^T(2)]^T,
\end{align*}
where $\underline x_{k}(2)\in\mathbb{R}^{\bar s_k}$, satisfying $\tilde D_{k}\underline x_{k}(2)= d_k$.
Besides, we define the  global state equality constraint set
\begin{align}\label{eq_space}
\mathbb{S}^{\bar s_k}\triangleq\{x\in\mathbb{R}^{\bar s_k}|\tilde D_{k}x= d_k\}.
\end{align}
Then define the constraint estimation error  
\begin{align}\label{def_error}
\underline e_{k,i}\triangleq \underline {\hat x}_{k,i}(2)-\underline x_{k}(2).
\end{align}
	Suppose $a_{ij,s}$ is the $(i,j)$th element of $\mathcal{A}^s$, where $\mathcal{A}$ is the adjacency matrix. Then, we obtain the following theorem, which shows the benefits of SEC and fusion-projection number $L$ to the dynamic estimation performance of TPDKF. 
	\begin{theorem}\label{prop_const} 
		\textbf{(Fusion-projection Compression)}
		Consider the system (\ref{system_all})-(\ref{system_constriants})  satisfying Assumption \ref{ass_topology}. For  TPDKF in Table \ref{ODKF2} with given fusion-projection number $L\geq 1$,  
				\begin{align}\label{eq_ineq}
				&E\{e_{k,i} e_{k,i}^T\}\\
				\leq &\left(\sum_{j\in\mathcal{V}}a_{ij,L}\tilde{P}_{k,j}^{-1}+\sum_{s=0}^{L-1}\left[\sum_{j\in \mathcal{V}}a_{ij,s}\frac{D_{k,j}^TD_{k,j}}{\varepsilon_j}\right]\right)^{-1}.\nonumber
				\end{align}
	Furthermore, $\forall\epsilon_0>0$, there exists an $L_0$, such that for $L\geq L_0,$
		\begin{align}\label{eq_small}
	\lambda_{max}\left(E\{\underline e_{k,i}\underline e_{k,i}^T\}\right)< \epsilon_0,
		\end{align}					
	where $e_{k,i}\triangleq\hat x_{k,i}-x_{k}$ is the estimation error of agent $i$, and $\underline e_{k,i}$ is the constraint estimation error defined in (\ref{def_error}).
	\end{theorem}	
	\begin{proof}
 Using the matrix inverse formula, one can obtain that $(\tilde P_{k,i}^{L})^{-1}=\sum_{j\in \mathcal{N}_{i}}a_{i,j}((\tilde P_{k,j}^{L-1})^{-1}+\frac{1}{\varepsilon_j}D_{k,j}^TD_{k,j})$.
 Recursively utilizing  the equation for $L-1$ steps, it then follows $(\tilde P_{k,i}^{L})^{-1}=\sum_{j\in \mathcal{V}}a_{ij,L}\tilde P_{k,j}^{-1}+S_{i,L}$,
 where $S_{i,L}=\sum_{s=0}^{L-1}\sum_{j\in \mathcal{V}}a_{ij,s}\frac{1}{\varepsilon_j}D_{k,j}^TD_{k,j}$,  $a_{ij,s}$ is the $(i,j)$th element of $\mathcal{A}^s,s\in\mathbb{Z}^+$, and  $a_{ij,0}=0,j\in \mathcal{N}_{i}^0,\, a_{ii,0}=1.$   
 According to the consistency in Lemma \ref{thm_pro}, $E\{e_{k,i} e_{k,i}^T\}=E\{\tilde e_{k,i}^L(\tilde e_{k,i}^L) ^T\}\leq \tilde P_{k,i}^{L}$. Thus, the conclusion in (\ref{eq_ineq}) holds.

    Given $\forall\epsilon_0>0$, we aim to prove the conclusion in (\ref{eq_small}). 
   According to the consistency in Lemma \ref{thm_pro}, 
   \begin{align}\label{eq_relax1}
   &F_{k}^{-1}E\{e_{k,i}e_{k,i}^T\}F_{k}^{-T}\leq F_{k}^{-1}\tilde P_{k,i}F_{k}^{-T}.
   \end{align}
  Considering $F_{k}^{-1} \tilde P_{k,i}F_{k}^{-T}$, it holds that $F_{k}^{-1} \tilde P_{k,i}F_{k}^{-T}\leq n diag\{\lambda_j,j=1,\cdots,n\}$, where $\lambda_j$ are the eigenvalues of  $F_{k}^{-1} \tilde P_{k,i}F_{k}^{-T}$.
    According to  Assumption \ref{ass_topology} and graph theory \cite{horn2012matrix, varga2009matrix}, $a_{ij,s}>0$ for $s\geq N-1$. Then, in light of $(\ref{eq_ineq})$, there exists an $L_0$, such that for 
    $L\geq L_0$, the eigenvalues of  $F_{k}^{-1} \tilde P_{k,i}F_{k}^{-T}$ corresponding to  the subspace related to constraint set $\mathbb{S}^{\bar s_k}$ are smaller than $\rho_0=\frac{\epsilon_0}{n}$. Thus,
 \begin{align}\label{eq_relax2}
 F^{-1} \tilde P_{k,i}F^{-T}
 \leq &  \left( \begin{array}{cc}
 n\lambda_{0}I_{n-s_k} & 0 \\
 0 & n\rho_0 \\
 \end{array}
 \right).
 \end{align}   
where $\lambda_{0}$  is the biggest eigenvalue of  $F_{k}^{-1} \tilde P_{k,i}F_{k}^{-T}$ in the rest subspace.
 According to (\ref{eq_relax1}) and (\ref{eq_relax2}),  
  the estimation error in the global constraint set  satisfies 
  		\begin{align}
  		\lambda_{max}\left(E\{\underline e_{k,i}\underline e_{k,i}^T\}\right)< n\rho_0=\epsilon_0.
  		\end{align}	
Then the conclusion in (\ref{eq_small}) holds.
	\end{proof}
According to Theorem \ref{prop_const}, the following two corollaries can be obtained.
\begin{corollary}\label{coro_1}
	Under Assumption \ref{ass_topology}, for the proposed TPDKF in Table \ref{ODKF2}, if $rank(\bar D_{k}^T\bar D_{k})=n$, then  $\forall\epsilon_0>0$, there exists an $L_0$, such that for $L\geq L_0,$					
	\begin{align}
\lambda_{max}\left(E\{e_{k,i} e_{k,i}^T\}\right)<\epsilon_0, \forall k\geq 0, \forall i\in\mathcal{V}.
	\end{align}
\end{corollary}
	
\begin{corollary}\label{coro_2}
		Under Assumption \ref{ass_topology}, for the proposed TPDKF in Table \ref{ODKF2},	the estimates $\hat{x}_{k,i} $ asymptotically satisfy the global constraints, i.e., 
	\begin{align}\label{prop_eq_pro}
	\lim\limits_{L\rightarrow +\infty} dis(\hat x_{k,i},P_{\mathcal{D}_{k}}[\hat x_{k,i}])=0, \forall k\geq 1,
	\end{align}
	where $\mathcal{D}_{k}=\bigcap_{i\in\mathcal{V}} \mathcal{D}_{k,i}$ and $\mathcal{D}_{k,i}=\{x_k|D_{k,i}x_k=d_{k,i}\}$.
\end{corollary}

	\begin{remark}
		Theorem \ref{prop_const} and Corollary \ref{coro_1} show the influence of fusion-projection number $L$ to the dynamic estimation performance of the filter in terms of upper bound of the error covariance. Corollary \ref{coro_2} show the global equality constraints
		will be asymptotically satisfied with the increasing of $L$.		
	\end{remark}
	\begin{remark}
		For $rank(\bar D_{k}^T\bar D_{k})=n$, all the elements of system state become deterministic, and the overall state can	be fully determined by the global SEC according to Corollary \ref{coro_1}. 
	\end{remark}

	The analysis  of upper boundedness of the error covariance matrix is one of the most difficult issues in the distributed estimation due to the time-varying dynamics and the strong correlation between the estimation errors of agents. In the following theorem, we provide the conclusion on mean square boundedness of estimation error under quite mild conditions.		
	\begin{theorem}\label{thm_consistent}
		(\textbf{Mean square boundedness}) 
		Consider the system (\ref{system_all})-(\ref{system_constriants}) satisfying Assumptions  \ref{ass_topology}-\ref{ass_observable}. For TPDKF, if  $L\geq 1$ and $\bar D_{k}=\bar D$ with rank $\bar s>0$, there exists a  matrix $ P_1>0$ and a matrix $P_2\geq 0$ with rank $\bar s$, such that for $\forall i\in \mathcal{V}$,
		\begin{equation}\label{thm_compare}
		\sup_{k\geq N+\bar N} E\{e_{k,i}e_{k,i}^T\}\leq  \left(P_1+LP_2\right)^{-1}<\infty.   
		\end{equation}
Furthermore,   there exists a scalar $\rho_1>0$, such that for $L\geq 1$, 
\begin{align}
	\lambda_{max}\left(\sup_{k\geq N+\bar N} E\{\underline e_{k,i}\underline e_{k,i}^T\}\right)< \frac{\rho_1}{L}, \forall i\in\mathcal{V},
\end{align}
where $e_{k,i}\triangleq\hat x_{k,i}-x_{k}$ is the estimation error of agent $i$, and  $\underline e_{k,i}$ is the constraint estimation error defined in (\ref{def_error}).
	\end{theorem}
	\begin{proof}
 Considering the consistency of TPDKF in Lemma \ref{thm_pro}, we turn to prove  the  upper boundedness of $P_{k,i}$ in TPDKF.	
		We consider $k\geq N+\bar N$. 
		By exploiting the matrix inverse formula on $P_{ k,i}$ and $\check P_{ k,i}$, respectively, we obtain
		\begin{align}\label{proof_stability13}
		P_{ k,i}^{-1}	=&\sum_{j\in \mathcal{V}}a_{ij,L}\tilde P_{k,j}^{-1}+\sum_{l=0}^{L-1}\sum_{j\in \mathcal{V}}a_{ij,l}\frac{D_{k,j}^TD_{k,j}}{\varepsilon_j}\\
		\geq&  \eta\sum_{j\in  \mathcal{V}}a_{ij,L}A_{ k-1}^{-T}P_{ k-1,j}^{-1} A_{ k-1}^{-1}\nonumber\\
		&+\sum_{j\in  \mathcal{V}} a_{ij,L}S_{ k,j}+\sum_{l=0}^{L-1}\sum_{j\in \mathcal{V}}a_{ij,l}\frac{D_{k,j}^TD_{k,j}}{\varepsilon_j},\nonumber
		\end{align}
		where $S_{k,i}=H_{ k,i}^TR_{k,i}^{-1}H_{ k,i}$, and
		$0< \eta<1$, which is obtained similarly to Lemma 1 in \cite{Battistelli2014Kullback} by noting the upper boundedness of   $Q_{k}$ and lower boundedness of $A_{k}A_{ k}^T$.
		By recursively applying  (\ref{proof_stability13}) for $\tilde L=N+\bar N$ times, one has
		\begin{align*}
		P_{ k,i}^{-1} 		\geq&\eta^{\tilde L}\Phi_{ k,k-\tilde L}^{-T}\left[\sum_{j\in \mathcal{V}} a_{ij,\tilde L\times L}P_{k-\tilde L,j}^{-1} \right] \Phi_{ k,k-\tilde L}^{-1}+
		\breve{P}_{ k,i}^{-1},
		\end{align*}
		where
			\begin{align}\label{proof_stability40}
			\breve{P}_{ k,i}^{-1}			=&
			\sum_{s=1}^{\tilde L}\eta^{s-1}\Phi_{ k, k-s+1}^{-T}
			\Bigg[
			\sum_{j\in \mathcal{V}} \big(a_{ij,s\times L}S_{ k-s+1,j}\\
			&\qquad+\sum_{l=0}^{L-1}a_{ij,s\times l }\frac{D_{ k-s+1,j}^TD_{ k-s+1,j}}{\varepsilon_j}
			\big)\Bigg] \Phi_{ k, k-s+1}^{-1},	\nonumber	
			\end{align}	
and $\Phi_{k,j}$ is the state transition matrix, 	$a_{ij,s}$ is the $(i,j)$th element of $\mathcal{A}^s$.
		According to  Assumption \ref{ass_topology} and graph theory \cite{horn2012matrix, varga2009matrix}, $a_{ij,s}>0$ for $s\geq N-1$.
		We consider $\breve{P}_{ k,i}^{-1}$.  From (\ref{proof_stability40}), one can obtain 
		\begin{align}\label{proof_stability4}
		\breve{P}_{ k,i}^{-1}
		\geq& a_{min}\eta^{\tilde L-1}\sum_{s=N}^{\tilde L}\Phi_{ k, k-s+1}^{-T}
		\Bigg[
		\sum_{j\in \mathcal{V}} \big(S_{ k-s+1,j}\nonumber\\
		&\qquad+L\frac{D_{ k-s+1,j}^TD_{ k-s+1,j}}{\varepsilon_j}
		\big)\Bigg] \Phi_{ k, k-s+1}^{-1},	\nonumber\\
		=& a_{min}\eta^{\tilde L-1}\sum_{s=N}^{\tilde L}\Phi_{ k, k-s+1}^{-T}
		\Bigg[H_{k-s+1}^T R_{k-s+1}^{-1}H_{k-s+1}\nonumber\\
		&+LD_{k-s+1}^T\Upsilon_{k-s+1}^{-1}D_{k-s+1}\Bigg] \Phi_{ k, k-s+1}^{-1},
		\end{align}	
		where $a_{min}=\min_{i,j\in \mathcal{V}}{a_{ij,s}>0,s\in [N:\tilde L\times L]}$ and
		$\Upsilon_{k}=blockdiag\{\varepsilon_1I_{s_{k,1}},\varepsilon_2I_{s_{k,2}},\ldots,\varepsilon_NI_{s_{k,N}}\}$.
		Under Assumption \ref{ass_observable}, for the given $\{\varepsilon_i,i\in\mathcal{V}\}$ and $\tilde L= N+\bar N$, there exists an $\bar \alpha>0$, such that
		\begin{equation*}
		\begin{split}
		&\sum_{s=N}^{\tilde L}\Phi_{ k, k-s+1}^{-T}
		\Bigg[H_{k-s+1}^T R_{k-s+1}^{-1}H_{k-s+1}\nonumber\\
		&+D_{k-s+1}^T\Upsilon_{k-s+1}^{-1}D_{k-s+1}\Bigg] \Phi_{ k, k-s+1}^{-1}\geq \bar\alpha I_{n}.
		\end{split}
		\end{equation*}

Denote $P_3=\bar\alpha a_{min}\eta^{N+\bar N-1}I_n$.		
Due to $0<\beta_2\leq A_{k}A_{k}^T\leq \beta_1 I_n$ and global constraint matrix $\bar D_{k}=\bar D$ with rank $\bar s$, then there exists a constant matrix $P_2$ with rank $\bar s$ subject to $0\leq P_2< P_3$, such that
\begin{align*}
\sum_{s=N}^{\tilde L}\Phi_{ k, k-s+1}^{-T}D_{k-s+1}^T\Upsilon_{k-s+1}^{-1}D_{k-s+1}\Phi_{ k, k-s+1}^{-1}\geq P_2\geq 0.
\end{align*}
Let $P_1=P_3-P_2> 0$. 	Therefore, for $L\geq 1$ and $k\geq N+\bar N$, 	$\breve{P}_{ k,i}^{-1}\geq P_3+(L-1)P_2\geq P_1+LP_2$. 
Considering the consistency and $P_{k,i}^{-1}\geq \breve{P}_{ k,i}^{-1}$, the conclusion of (\ref{thm_compare}) holds.

Due to $\bar D_{k}=\bar D$, for the global constraint matrix $\bar D\in\mathbb{R}^{\bar s\times\bar s}$, it holds that $\bar D\bar D^T\geq \varpi I_{\bar s}$.
Thus, from (\ref{proof_stability4}) and  we have
\begin{align}\label{eq_relax3}
\breve{P}_{ k,i}^{-1}\geq a_{min}\eta^{\tilde L-1} \left(\bar\alpha I_{n}+(L-1)P(\varpi)\right).
\end{align}
where $P(\varpi)$ is a diagonal matrix consisting of $n-\bar s$ zeros and $\bar s$ elements of $\varpi$.
Considering the state constraint set $\mathbb{S}^{\bar s_k}=\mathbb{S}^{\bar s}$ defined in (\ref{eq_space}),
it is similar as the proof of Theorem \ref{prop_const} to obtain
$E\{\underline e_{k,i}\underline e_{k,i}^T\}\leq \frac{\rho_1}{L},$ $\forall k\geq  N+\bar N$, where $\rho_1$ is easily obtained through (\ref{eq_relax3}).
	\end{proof}
\begin{remark}
	1) 
	Theorem \ref{thm_consistent} shows the SEC can  relax either the collective observability assumptions given in \cite{Battistelli2014Kullback,Battistelli2015Consensus,Battistelli2016stability,Dm2015Distributed,wang2017convergence} or the local observability conditions given in \cite{olfati2009kalman,stankovic2009consensus,yang2014stochastic}.
		
	2) Theorem \ref{thm_consistent} reveals the influence of fusion-projection number $L$ to the  boundedness of error covariance and shows that the bound of the mean square estimation error in the constraint set can be  arbitrarily small by setting a sufficiently large fusion-projection number $L$.
\end{remark}
For the initial biased estimates resulted from information asymmetry of agents,  we can obtain Theorem \ref{thm_noisefree} which depicts the asymptotic unbiasedness of state estimate for the proposed TPDKF. The following lemma on matrix approximation will help to prove Theorem \ref{thm_noisefree}.
\begin{lemma}\label{lem_ineq}
	Suppose $\varPi_{0}$ and $\varPi_{1}$ satisfy  $0\leq\varPi_{0}\leq \varPi_{1}$ and $\varPi_{1}>0$, then  $\varPi_{0}\varPi_{1}^{-1}\varPi_{0}\leq\varPi_{0}$.
\end{lemma}
\begin{proof}
	We utilize a matrix approximation method to finish the proof.
	Define $\varPi=\varPi_{0}+\lambda(\varPi_{1}-\varPi_{0})$, where $0<\lambda<1$. Then under the conditions of this lemma, it is easy to obtain that $0<\varPi\leq \varPi_{1}$. Hence, we have $\varPi_{1}^{-1}\leq \varPi^{-1}$, which leads to $\varPi\varPi_{1}^{-1}\varPi\leq \varPi$. Let $\lambda\rightarrow 0^+$, and we  obtain 
$\varPi_{0}\varPi_{1}^{-1}\varPi_{0}\leq\varPi_{0}$.	
\end{proof}

\begin{theorem}\label{thm_noisefree}
			(\textbf{Asymptotically unbiased}) Consider the system (\ref{system_all})-(\ref{system_constriants}) satisfying Assumptions  \ref{ass_topology}-\ref{ass_observable}. For TPDKF, if the fusion-projection number $L\geq1$, 
the state estimate is asymptotically unbiased with an exponential rate, i.e., there exist constants $b>0$ and $0<\varrho<1$, such that
	\begin{equation}\label{zero_covergence}
	\|E\{\hat x_{k,i}-x_{k}\}\|_2\leq b\varrho^k,  
	\end{equation}
	where $ k\in\mathbb{N}, i\in \mathcal{V}.$
\end{theorem}
\begin{proof}
	Recall the denotations of estimation errors below Table \ref{filter_stru}.
	Consider the function $V_{k,i}(E\{\bar e_{k,i}\})=E\{\bar e_{k,i}\}^T\bar P_{k,i}^{-1}E\{\bar e_{k,i}\}$.
	According to the fact (\expandafter{\romannumeral3}) of Lemma 1 in \cite{Battistelli2014Kullback} and  $0<\underline Q\leq Q_{k}\leq \overline Q<+\infty$, there exists a $\varrho\in(0,1)$, such that
	\begin{equation}\label{iteration_V1}
	\begin{split}
	V_{k,i}(E\{\bar e_{k,i}\})\leq& \varrho\bar E\{\bar e_{k,i}\}^TA_{k}^{-T}P_{k,i}^{-1}A_{k}^{-1}E\{\bar e_{k,i}\}\\
	\leq& \varrho E\{ e_{k,i}\}^TP_{k,i}^{-1}E\{ e_{k,i}\}.
	\end{split}
	\end{equation}

	Then $E\{\tilde e_{k,i}^{l+1}\}=P_{k,i}^{D,l}\check P_{k,i}^l\sum_{j\in \mathcal{N}_{i}}a_{i,j}(\tilde P_{k,j}^l)^{-1}E\{\tilde e_{k,j}^{l}\},$
where $P_{k,i}^{D,l}= I_n-\check P_{k,i}^lD_{k,i}^T(D_{k,i}\check P_{k,i}^lD_{k,i}^T)^{-}D_{k,i}$. Employing the matrix approximation method in Lemma \ref{lem_ineq} yields $P_{k,i}^{D,l}\check P_{k,i}^l(\tilde P_{k,i}^{l+1})^{-1}P_{k,i}^{D,l}\check P_{k,i}^l\leq P_{k,i}^{D,l}\check P_{k,i}^l\leq \check P_{k,i}^l$.
Applying  Lemma 2 in \cite{Battistelli2014Kullback} to the right hand of  (\ref{iteration_V1}), one can obtain that
$E\{(\tilde e_{k,i}^{l+1})\}^T(\tilde P_{k,i}^{l+1})^{-1}E\{(\tilde e_{k,i}^{l+1})\}\leq \sum_{j\in \mathcal{N}_{i}}a_{i,j}E\{(\tilde e_{k,j}^{l})\}^T(\tilde P_{k,j}^{l})^{-1}E\{(\tilde e_{k,j}^{l})\}$.
By applying this iteration for $l=0,1,\cdots,L-1$, we have
\begin{align}\label{eq_recur2}
&E\{e_{k,i}\}^TP_{k,i}^{-1}E\{e_{k,i}\}\nonumber\\
\leq& \sum_{j\in \mathcal{V}}a_{ij,L}E\{(\tilde e_{k,j}^{0})\}^T(\tilde P_{k,j}^{0})^{-1}E\{(\tilde e_{k,j}^{0})\}\nonumber\\
=&\sum_{j\in \mathcal{V}}a_{ij,L}E\{\tilde e_{k,j}\}^T\tilde P_{k,j}^{-1}E\{\tilde e_{k,j}\}.
\end{align}
	Since $\tilde P_{k,i}=(I-K_{k,i}H_{k,i})\bar P_{k,i}$, one can get $\tilde P_{k,i}^{-1}(I-K_{k,i}H_{k,i})=\bar P_{k,i}^{-1}$.
	From TPDKF, we have $\check P_{k,i}=(\sum_{j\in \mathcal{N}_{i}}a_{i,j}\tilde P_{k,j}^{-1})^{-1}\leq (\sum_{j\in \mathcal{N}_{i}}a_{i,j}\bar P_{k,j}^{-1})^{-1}$. Considering $E\{\tilde e_{k,i}\}=(I-K_{k,i}H_{k,i})E\{\bar e_{k,i}\}$, (\ref{iteration_V1})  and (\ref{eq_recur2}), it can be obtained that
	$V_{k+1,i}(E\{\bar e_{k+1,i}\})\leq\varrho\sum_{j\in \mathcal{V}}a_{ij,L}V_{k,j}(E\{\bar e_{k+1,i}\}).$
	Then, we have 
			\begin{align}\label{V_ineq_final6}
			V_{k+1,i}(E\{\bar e_{k+1}\})&\leq \varrho^{k+1}\sum_{j\in \mathcal{V}}a_{ij,L(k+1-k^*)}V_{0,j}(E\{\bar e_{0,j}\})\nonumber\\
			&\leq \varrho^{k+1}V_{0,max}, \quad 0<\varrho<1,
			\end{align}
		where $V_{0,max}=\max\limits_{j\in\mathcal{V}}V_{0,j}(E\{\bar e_{0,j}\})>0$.	
	Under the conditions of this theorem, the conclusion of Theorem \ref{thm_consistent} holds. 
	Then  $\bar P_{k,i}$ is uniformly upper bounded due to $0<\beta_{2}I_n\leq A_{k}A_{k}^T\leq\beta_{1}I_n$ and  $Q_{k}\leq \overline Q<+\infty$. 
	Since $\bar P_{k+1,i}$ is uniformly upper bounded, i.e., there exists a constant matrices $\bar P_1>0$, such that $\bar P_{k+1,i}\leq \bar P_1.$
	Therefore,
	\begin{align}\label{ineq_1}
	V_{k+1,i}(E\{\bar e_{k+1,i}\})
	=& E\{\bar e_{k+1,i}\}^T\bar P_{k+1,i}^{-1}E\{\bar e_{k+1,i}\}\nonumber\\
	\geq& \lambda_{min}(\bar P_{1}^{-1})\|E\{\bar e_{k+1,i}\}\|_2.
	\end{align}	
	Substituting (\ref{ineq_1}) into (\ref{V_ineq_final6}) and considering  $E\{\bar e_{k+1,i}\}=A_{k}E\{e_{k,i}\}$, we can obtain
	$\|E\{e_{k,i}\}\|_2\leq \frac{ \varrho V_{0,max}}{\lambda_{min}(A_{k}^TA_{k})\lambda_{min}(\bar P_{1}^{-1})}\varrho^{k}.$
	Due to $0<\beta_{2}I_n\leq A_{k}A_{k}^T$. Let $b=\frac{\varrho V_{0,max}}{\beta_{2}\lambda_{min}(\bar P_{1}^{-1})}>0$. Then  we have $\|E\{e_{k,i}\}\|_2\leq  b\varrho^k.$ Then for $k\in\mathbb{N}$,
	$\|E\{e_{k,i}\}\|_2\leq b\varrho^k$.
\end{proof}	
\begin{corollary}\label{cor_pro}
	{\
	Under the same conditions as Theorem \ref{thm_noisefree}, $\lim\limits_{k\rightarrow +\infty} dis(E\{\hat x_{k,i}\},P_{\mathcal{D}_{k}}[E\{\hat x_{k,i}\}])=0$ with an exponential rate, where $\mathcal{D}_{k}=\bigcap_{i\in\mathcal{V}} \mathcal{D}_{k,i}$ and $\mathcal{D}_{k,i}=\{x_k|D_{k,i}x_k=d_{k,i}\}$.
	}
\end{corollary}
\begin{proof}
	Because $dis(P_{\mathcal{D}}[x],P_{\mathcal{D}}[y])\leq dis(x,y)$ \cite{lou2013convergence} and  $x_{k}\in\mathcal{D}_{k}$, 
	\begin{align}
	&dis(E\{\hat x_{k,i}\},P_{\mathcal{D}_{k}}[E\{\hat x_{k,i}\}])\nonumber\\
	\leq& dis(E\{\hat x_{k,i}\},E\{x_k\})+dis(E\{x_k\},P_{\mathcal{D}_{k}}[E\{\hat x_{k,i}\}])\nonumber\\
	= &dis(E\{\hat x_{k,i}\},E\{x_k\})+dis(P_{\mathcal{D}_{k}}[E\{x_k\}],P_{\mathcal{D}_{k}}[E\{\hat x_{k,i}\}])\nonumber\\
	\leq& 2dis(E\{\hat x_{k,i}\},E\{x_k\}).\label{eq_cor}
	\end{align}
According to Theorem \ref{thm_noisefree}, $dis(E\{\hat x_{k,i}\},E\{x_k\})\rightarrow 0$ as $k\rightarrow\infty$ with an exponential rate. Therefore, the conclusion holds due to (\ref{eq_cor}).
\end{proof}

\begin{remark}
Theorem \ref{thm_noisefree}  shows the proposed filter can remove the  estimation bias resulted from initial information asymmetry at an exponentially fast rate. 
	It can be seen from Corollary \ref{cor_pro} that the mean of state estimate for each agent will eventually satisfy global  SEC.
\end{remark}

	\section{Distributed Filter with Event-triggered communications}
	In the commonly-used time-based distributed communication mechanism, the local messages of each agent may be broadcasted to its neighbors for several times between two measurement updates.
	Such a communication strategy may lead to some unnecessary communication data flow in the network, which unavoidably increases the communication  burden and energy consumption.
	In fact, it is quite necessary to employ efficient strategies to reduce communication rate and save energy, because of the practical bandwidth and energy constraints (e.g., in wireless sensor network).

	Given fusion-projection step $L\geq 1$ in Table \ref{ODKF2}, we aim to reduce the communication frequencies of the agents not sending sufficiently new information. As a result, different agents over the network may have different communication times  ranging from zero to $L$.
	In this section, for convenience, we study the fusion-projection step $L=1$. For $L>1$, the similar results can be obtained.
	To study event-triggered schemes, we focus on time-invariant systems, though the method can be extended to time-varying cases.  
	In other words, we give the following assumption for simplicity.
	
	\begin{assumption}\label{ass_A2}
		The matrices in the system (\ref{system_all})-(\ref{system_all1}) and the constraints (\ref{system_constriants}) satisfy $A_{k}=A$,  $Q_{k}=Q,H_{k,i}=H_{i},R_{k,i}=R_{i},D_{k,i}=D_{i},d_{k,i}=d_i,\forall i\in \mathcal{V}, k=1,2,\ldots$.
	\end{assumption}
	
	Then the following two lemmas are quite straightforward.
	
	\begin{lemma}
		Under Assumptions \ref{ass_observable} and \ref{ass_A2}, $[A,\bar H]$ is observable, where
		$\bar H=[H^T,D^T]^T, H= [H_{1}^T,H_{2}^T,\ldots,H_{N}^T]^T, D=[D_{1}^T,D_{2}^T,\ldots,D_{N}^T]^T.$
	\end{lemma}
	
	\begin{lemma}
		Under Assumption  \ref{ass_A2}, the system matrix $A$ is non-singular.	
	\end{lemma}
	
	Different from many existing results with event-triggered schemes \cite{liu2015event}, {whose  assumptions on the upper boundedness of the error covariance matrix  are usually  related to the existence of the solutions of Riccati equations or Hurwitz stability of $A$,}
	the collective observability of $[A,\bar H]$ of the system can be verified before the implementation of algorithms.
	Clearly, the collective observability of $[A,\bar H]$ is a time-invariant version of the ECO condition in Assumption \ref{ass_observable}, which is weaker than the assumption that $(A,H)$ is observable given in \cite{Battistelli2014Kullback,Battistelli2016Distributed,Battistelli2015Consensus}.  The non-singularity of $A$ can be guaranteed through discretization from general continuous systems.
	
	\subsection{Event-triggered Communication Scheme}
	In this subsection, we will study a event-triggered communication scheme for distributed Kalman filter.
	{For  TPDKF in Table \ref{ODKF2} with $L=1$, each agent can decide whether to broadcast the message $(\tilde x_{k,j},\tilde P_{k,j})$ to its out-neighbors or not, thus we need to give a criterion to determine what message is worth broadcasting.}
	Here, we utilize the SoD regulation triggering mechanism, which is an event-triggered principle and related to the information increment.
	Different from the SoD methods in \cite{liu2015event,andren2016event} depending on some stochastic variables, a novel  SoD method
	is given in this paper. The event triggers of proposed SoD method is indeed deterministic and  based on $\tilde P_{k,j}^{-1}$, where one can judge the information increment according to its variation degree.

	Let the pair ($\tilde x_{t,j},\tilde P_{t,j}$) be the latest  message broadcasted by agent $j$ to its out-neighbors.
	Define the following triggering mechanism function for agent $j$ as $g_{k,j}(\cdot):\mathbb{R}^{n\times n}\times \mathbb{R}^{n\times n}\times \mathbb{R}\rightarrow \mathbb{R}$:
	\begin{equation}\label{triggering}
	g_{k,j}(\tilde P_{k,j},\bar{\tilde{P}}_{k,j},\delta_{j})=\lambda_{max}\big(\tilde P_{k,j}^{-1}-\bar{\tilde{P}}_{k,j}^{-1}\big)-\delta_{j},j\in\mathcal{V},
	\end{equation}
	where  $\delta_{j}\geq0$ is the triggering threshold of agent $j$ which is usually predefined, and
	\begin{equation}\label{eq_pred_matrix}
	\bar{\tilde{P}}_{k,j}\triangleq A^{k-t}\tilde P_{t,j} (A^{k-t})^T+\sum_{l=t}^{k-1} A^{k-1-l}Q(A^{k-1-l})^T
	\end{equation}
	is the multi-step prediction matrix.

	The event for agent $j$ at time $k$ is triggered if $ g_{k,j}> 0$. In other words, if $ g_{k,j}> 0$, agent $j$ broadcasts its information message to its out-neighbors at time $k$.
	Suppose that the event is triggered at the initial time.  Define the pair $(\tilde x_{k,j}^t,\tilde P_{k,j}^t)$ as
	\begin{equation}\label{estimates_lattest}
	(\tilde x_{k,j}^t,\tilde P_{k,j}^t)=
	\begin{cases}
	\big(A^{k-t}\tilde x_{t,j},\bar{\tilde{P}}_{k,j}\big),&\text{if}\quad  g_{k,j}\leq 0,\\
	\big(\tilde x_{k,j},\tilde P_{k,j}\big),&\text{if}\quad  g_{k,j}> 0.
	\end{cases}
	\end{equation}
	If $ g_{k,j}>0$, the event for agent $j$ is triggered at this moment and $\big(\tilde x_{k,j},\tilde P_{k,j}\big)$ is broadcasted to the out-neighbors of agent $j$, and then each out-neighbor receives $(\tilde x_{k,j}^t,\tilde P_{k,j}^t)=\big(\tilde x_{k,j},\tilde P_{k,j}\big)$. Otherwise, these out-neighbors obtain no information of agent $j$ at this moment and they make a multi-step prediction using the latest received pair  ($\tilde x_{t,j},\tilde P_{t,j}$) to obtain $(\tilde x_{k,j}^t,\tilde P_{k,j}^t)=\big(A^{k-t}\tilde x_{t,j},\bar{\tilde{P}}_{k,j}\big)$.
	Based on the above discussion,  we propose an event-triggered projected distributed Kalman filter (EPDKF) in Table \ref{ODKF3}.
	\begin{table}[htp]
		\caption{Event-triggered Projected Distributed Kalman Filter (EPDKF)}
		\label{ODKF3}	
		\centering  
		\begin{tabular}{l}  
			\hline\hline
			\textbf{Initialization:}\\
			{
			$(\hat x_{0,i},P_{0,i})$ is satisfied with $P_{0,i}\geq (1+\theta_i)P_0+\frac{\theta_i+1}{\theta_i}P_{0,i}^*$,$\exists \theta_i>0, $}\\
			{where $P_{0,i}^*\geq (\hat x_{0,i}-E\{x_0\})(\hat x_{0,i}-E\{x_0\})^T$,}\\
			\textbf{Input:}\\	
			$(\hat x_{k-1,i},P_{k-1,i},\varepsilon_i,\delta_i)$,  \\
			\textbf{Prediction:}\\
			$\bar x_{k,i}=A\hat x_{k-1,i},$\\  
			$\bar P_{k,i}=AP_{k-1,i}A^T+Q,$\\         
			\textbf{Measurement Update:}\\
			$\tilde x_{k,i}=\bar x_{k,i}+K_{k,i}(y_{k,i}-H_{i}\bar x_{k,i}),$\\        
			$K_{k,i}=\bar P_{k,i}H_{i}^T(H_{i}\bar P_{k,i}H_{i}^T+R_{i})^{-1},$\\
			$\tilde P_{k,i}=(I-K_{k,i}H_{i})\bar P_{k,i}$,\\
			\textbf{Local Fusion:}Obtain the pair ($\tilde x_{k,j}^t$, $\tilde P_{k,j}^t$ ) through  (\ref{estimates_lattest})\\
			$\check x_{k,i}=\check P_{k,i}\bigg(a_{i,i}\tilde P_{k,i}^{-1}\tilde x_{k,i}+\sum_{j\in \mathcal{N}_{i}^0}a_{i,j} (\tilde P_{k,j}^t)^{-1}\tilde x_{k,j}^t\bigg)$,\\	
			$\check P_{k,i}=\bigg(a_{i,i}\tilde P_{k,i}^{-1}+\sum_{j\in \mathcal{N}_{i}^0}a_{i,j}  (\tilde P_{k,j}^t)^{-1}\bigg)^{-1}$,	\\
			\textbf{Projection:}\\
			$\hat x_{k,i}=\check x_{k,i}-\check P_{k,i}D_{i}^T(D_{i}\check P_{k,i}D_{i}^T)^{-}(D_{i}\check x_{k,i}-d_{i}),$\\
			$P_{k,i}=\check P_{k,i}-\check P_{k,i}D_{i}^T(D_{i}\check P_{k,i}D_{i}^T+\varepsilon_i I_{s_i})^{-1}D_{i}\check P_{k,i}$,\\	
			\textbf{Output:}\\	
			$(\hat x_{k,i},P_{k,i})$.	\\
			\hline
		\end{tabular}
	\end{table}
	
	Clearly, given threshold $\delta_{i}$, the design of the filtering gain, the parameter matrices and the event-triggered scheme, simply depends  on the local available information without the global knowledge of the system or the network topology.  Hence, the proposed algorithm EPDKF is a distributed filtering algorithm.

	\subsection{Estimation Performance of EPDKF}\label{subsec_EPDKF}
	In this subsection, we will investigate the main estimation performance of EPDKF and provide  design principle for triggering thresholds.
	The following result shows the Gaussian distribution of estimation error for EPDKF in Table \ref{ODKF3}.	
	\begin{proposition}\label{thm_distri2}
	{		Consider the system (\ref{system_all})--(\ref{system_all1}) with constraints (\ref{system_constriants}). For EPDKF, the estimation error $e_{k,i}=\hat x_{k,i}-x_{k}$ is Gaussian, $\forall i\in \mathcal{V}, k=1,2,\ldots$.}
	\end{proposition}	
	\begin{proof}
		The proof is similar to the proof of Proposition \ref{thm_distri}.
	\end{proof}	
	On the state estimates of EPDKF, we have the following lemma.	
	\begin{lemma}\label{thm_pro2}
	For EPDKF of the system (\ref{system_all})--(\ref{system_all1}) with constraints (\ref{system_constriants}),
	the state estimate $\hat x_{k,i}$  satisfies the SEC of agent $i$ at time $k$, and 
	 the pair $(\hat x_{k,i}, P_{k,i})$ is consistent.   Moreover, $ P_{k,i}\leq\check P_{k,i}, \forall i\in \mathcal{V}, k=1,2,\ldots$.
	\end{lemma}
	\begin{proof}
		See the proof in Appendix \ref{pf_thm4}.
	\end{proof}

	Like TPDKF, the estimation error covariances of EPDKF in the stages of prediction, measurement update and local fusion are also upper bounded by  $\bar P_{k,i}$, $\tilde P_{k,i}$ and $\check P_{k,i}$.
	Thus, these parameter matrices can be employed to evaluate the estimation error in real time under the error distribution illustrated in Proposition \ref{thm_distri2}.  Additionally, for EPDKF,  the inverse of $\tilde P_{k,i}$  can be treated as a lower bound of information matrix, which contributes much in the design of event-triggered mechanism (\ref{estimates_lattest}). 

	{
	The following result shows  upper boundedness of the error covariance matrix of the proposed EPDKF algorithm.}
	
	\begin{theorem}\label{thm_consistent2}
		(\textbf{Mean square boundedness}) For EPDKF in Table \ref{ODKF3},	under Assumptions \ref{ass_topology} -- \ref{ass_A2},
		there exists a positive definite matrix $\hat P$, such that
		\begin{equation}\label{thm_compare22}
		0<E\{e_{k,i}e_{k,i}^T\}\leq \hat P<\infty, \quad  \forall k\geq 0, i\in \mathcal{V},
		\end{equation}
		if 	the triggering thresholds $\{\delta_{j},j\in \mathcal{V}\}$ satisfy
		\begin{equation}\label{delta_condition120}
		\begin{split}
		&\sum_{j\in \mathcal{V}}\delta_{j}M_{i,j}
		<\bar M_{i}, \forall  i\in \mathcal{V},
		\end{split}
		\end{equation}
		where
		\begin{equation}\label{delta_condition1200}
						\left\{
		\begin{aligned}
&M_{i,j}=\sum_{\tau=1}^{k^{*}}\beta^{\tau-1}a_{ij,\tau}(A^{1-\tau})^{T}A^{1-\tau}, \\
&\bar M_{i}=\sum_{\tau=1}^{k^{*}}\beta^{\tau-1}(A^{1-\tau})^{T}\tilde M_{i,j,\tau}A^{1-\tau}\\
&\tilde M_{i,j,\tau}=\sum_{j\in \mathcal{V}}(a_{ij,\tau}H_{j}^TR_{j}^{-1}H_{j}+\frac{a_{ij,\tau-1}}{\varepsilon_j }D_{j}^TD_{j}),
		\end{aligned}
		\right.
		\end{equation}
		for all $k^{*}\geq N+n$, $\beta$ given in  (\ref{P_recursive}) and
		$a_{ij,\tau}$ is the $(i,j)$th element of $\mathcal{A}^\tau$.
	\end{theorem}
	\begin{proof}
		Considering the consistency of EPDKF in Lemma \ref{thm_pro2}, we turn to prove  the  upper boundedness of $P_{k,i}$ in EPDKF.
		Given EPDKF and the triggering condition  (\ref{estimates_lattest}),  exploiting  matrix inverse formula on $P_{k,i}$ and $\check P_{k,i}$ yields
		\begin{align}\label{proof_stability213}
		P_{k,i}^{-1}&\geq a_{i,i}\tilde P_{k,i}^{-1}+\sum_{j\in \mathcal{N}_{i}^0}a_{i,j}  (\tilde P_{k,j}^t)^{-1}+\frac{1}{\varepsilon_i }D_{i}^TD_{i}\nonumber\\
		&\geq\sum_{j\in \mathcal{N}_{i}}a_{i,j} (\bar P_{k,j}^{-1}+H_{j}^TR_{j}^{-1}H_{j}-\delta_{j})+\frac{1}{\varepsilon_i }D_{i}^TD_{i}.
		\end{align}	
		By Lemma 1 in \cite{Battistelli2014Kullback}, 	there exists a real scalar $\beta\in(0,1)$ such that
		\begin{align}\label{P_recursive}
		P_{k,i}^{-1}
		&\geq\beta\sum_{j\in \mathcal{N}_{i}}a_{i,j}A^{-T}P_{k-1,j}^{-1}A^{-1}\\
		&\quad+\sum_{j\in \mathcal{N}_{i}}a_{i,j}(H_{j}^TR_{j}^{-1}H_{j}-\delta_{j})+\frac{1}{\varepsilon_i }D_{i}^TD_{i}.\nonumber
		\end{align}	
		By recursively applying the above inequality for $k^{*}$ times, we obtain
		\begin{equation*}\label{P_recursive2}
		\begin{split}
		P_{k,i}^{-1}\geq&\beta^{k^{*}}\sum_{j\in \mathcal{V}}a_{ij,k^{*}}\left[ (A^{-k^{*}})^{T}P_{k-k^{*},j}^{-1}A^{-k^{*}}\right] +\hat P_{i}^{-1}.
		\end{split}
		\end{equation*}	
		where
		\begin{align}
		\hat P_{i}^{-1}=&\sum_{\tau=1}^{k^{*}}\beta^{\tau-1}\sum_{j\in \mathcal{V}}(A^{1-\tau})^{T}\big(a_{ij,\tau}H_{j}^TR_{j}^{-1}H_{j}\nonumber\\
		&+\frac{a_{ij,\tau-1}}{\varepsilon_j }D_{j}^TD_{j}-a_{ij,\tau}\delta_{j}\big)A^{1-\tau}.\label{P_upper}
		\end{align}

If $\hat P_i^{-1}$ is positive definite, we have $P_{k,i}\leq \hat P_i$. Thus, there exists a matrix $P^*$ such that $P_{k,i}\leq P^*, \forall i\in \mathcal{V}$.
		To guarantee the positiveness of $\hat P_i^{-1}$, $\delta_{j},j\in \mathcal{V}$ can be designed such that
		$\sum_{\tau=1}^{k^{*}}\beta^{\tau-1}\sum_{j\in \mathcal{V}}a_{ij,\tau}(A^{1-\tau})^{T}\delta_{j}A^{1-\tau}
		<\check{P}_{i}^{-1}.$		
		where $\check{P}_{i}^{-1}=\sum_{\tau=1}^{k^{*}}\beta^{\tau-1}\sum_{j\in \mathcal{V}}(A^{1-\tau})^{T}(a_{ij,\tau}H_{j}^TR_{j}^{-1}H_{j}+\frac{a_{ij,\tau-1}}{\varepsilon_j }D_{j}^TD_{j})A^{1-\tau}.$
		Then the conclusion in  (\ref{delta_condition120}) is reached.
		For the existence of a positive upper bound of $\delta_{j},j\in\mathcal{V}$,  one needs to show the positiveness of $\check{P}_{i}^{-1}$. 	
		According to  Assumption \ref{ass_topology} and graph theory \cite{horn2012matrix, varga2009matrix}, we have $a_{ij,s}>0$, $s\geq N-1$. Thus,
		$\check{P}_{i}^{-1}	\geq a_{min}\beta^{k^{*}-1}(A^{-N})^{T}D_{H}A^{-N},$
		where
		\begin{align*}
		\begin{cases}
		R= blockdiag\{R_{1},\ldots,R_{N}\}\\
		a_{min}=\min{a_{ij,\tau}>0,\tau\in [N:k^{*}]}\\
		D_{H}=\sum_{l=0}^{k^{*}-N-1}(A^{-l})^{T}(H^TR^{-1}H+D^T\Upsilon^{-1}D)A^{-l}.
		\end{cases}		
		\end{align*}

		Consider the observability matrix  \\$ O_{n}=\sum_{t=0}^{n-1}(A^t)^{T}(H^TR^{-1}H+D^T\Upsilon^{-1}D)A^{t}=G_{n}^T \bar R_{n}^{-1}G_{n},$
		where $\bar R_{n}=I_{n}\otimes blockdiag\{R,\Upsilon\}$,  and
		$	G_{n}=\left[H^T,D^T,\ldots,(HA^{n-1})^T,(DA^{n-1})^T \right]^T$.
		Under Assumption \ref{ass_observable}, it easily to see that $O_{n}>0$ and $G_{n}$ is of column full rank.
		Define the matrix $F_{n}=A^{1-n}$, and then the matrix $G_{n}F_{n}$ is still of column full rank, where
		$G_{n}F_{n}=\left[(HA^{1-n})^T ,(DA^{1-n})^T  ,\ldots, H^T,D^T \right]^T.  $	
		Thus, $F_{n}^TO_{n}F_{n}=F_{n}^TG_{n}^T \bar R_{n}^{-1}G_{n}F_{n}>0$.	
		Considering the form $F_{n}^TO_{n}F_{n}$ and $D_H$,
		let $k^{*}\geq N+n$, then the positiveness of $\check{P}_{i}^{-1}$ is verified.
		For $0\leq k\leq k^{*}$, there exists a sufficiently large $P_{*}$ such that $0<P_{k,i}\leq P_{*}, \quad \forall i\in \mathcal{V}, \forall k=0,1,2,\ldots,k^{*}.$ Recall that $P_{k,i}\leq P^*, k>k^*,\forall i\in \mathcal{V}$.
		Thus, under  condition (\ref{delta_condition120}), it is straightforward to guarantee  (\ref{thm_compare22}).	
	\end{proof}

	{
	\begin{corollary}\label{cor1}
		Under the same conditions as  Theorem \ref{thm_consistent2} and  $\delta_{j}=\delta,\; j\in \mathcal{V}$, if
		\begin{equation}\label{design_interval}
		\delta
		<\lambda_{min}(M_{i}^{-\frac{1}{2}}\bar M_{i}M_{i}^{-\frac{1}{2}}), \forall  i\in \mathcal{V},
		\end{equation}	
		\item 1) there exists a positive definite matrix $\hat P_{\delta}$ such that $0<P_{k,i}\leq \hat P_{\delta}<\infty, \quad  \forall k\geq 0, \forall i\in \mathcal{V}.$
		\item 2) $\hat P_{\delta}$ is non-decreasing with respect to $\delta$, i,e., $P_{\delta_1}\leq P_{\delta_2},\delta_1\leq \delta_2,$
			with $M_{i}=\sum_{j\in \mathcal{V}}M_{i,j}$, where $M_{i,j}$ and $\bar M_{i}$ are defined in  (\ref{delta_condition1200}).
	\end{corollary}
	\begin{proof}
		The first conclusion is directly obtained from Theorem \ref{thm_consistent2}. Considering (\ref{P_upper}), the second conclusion can be easily proved.
	\end{proof}
}

	{
	\begin{remark}
			The conditions (\ref{delta_condition120}) and (\ref{design_interval}) essentially provide design principles for  the triggering thresholds $\delta_{i},i\in\mathcal{V}$ to guarantee the mean square boundedness of estimation error.  Although the upper bounds (\ref{delta_condition120}) and (\ref{design_interval}), which the triggering conditions depend on, are related to the network topology and overall system, through centralized design, they can be checked  before the implementations of the filters by each agent.
	\end{remark}
}
	\begin{remark}	
		Under collective observability conditions, $\delta_{i},i\in\mathcal{V}$ cannot be set too large.  Otherwise, suppose $\delta_{i},i\in\mathcal{V}$ are sufficiently large. Then, agents will not communicate with each other, which may lead to instability of estimates under collective observability conditions.
	\end{remark}

\subsection{Communication Rate}
In the subsection \ref{subsec_EPDKF}, we have analyzed the main estimation performance of the proposed EPDKF with event-triggered communications.
Another essential aspect, namely, the communication rate of event-triggered scheme, will be studied in this subsection. 
In the following, for convenience of analysis, we suppose $\delta_{j}=\delta,\; j\in \mathcal{V}$, which means all agents share the same triggering threshold. The conclusion for $\delta_{j}\neq\delta,j\in \mathcal{V},$ can be obtained similarly. Suppose there is at most one communication between two updates for all agents, then we provide the definition of the communication rate.

\begin{definition}
	For a multi-agent system, 	 the \textbf{communication rate} $\lambda$ of a distributed filter is defined as
\begin{align}\label{eq_rate}
\lambda=1-\frac{\sum_{i\in\mathcal{V}}p_i|\mathcal{N}_{i,out}^0| }{\sum_{i\in\mathcal{V}}|\mathcal{N}_{i,out}^0|},
\end{align}
where $p_i\in[0,1]$ is the ratio of non-triggering times for agent $i$ in a time interval of interest, $|\mathcal{N}_{i,out}^0|$ is the out-degree of agent $i$ without including itself, 	 and $\mathcal{V}$ stands for the index set of all agents.
\end{definition}

From (\ref{eq_rate}), we see that the communication rate $\lambda$ is directly influenced by the ratio $p_i$ of the non-triggering event. If the triggering condition in (\ref{estimates_lattest}) is satisfied in the entire interval for all agents, i.e., $p_i=0$, then $\lambda=1$. If the triggering conditions in (\ref{estimates_lattest}) are not satisfied at any time for any agent, i.e., $p_i=1$, then $\lambda=0$.
Therefore, the definition of $\lambda$ is reasonable  to represent the communication rate during the entire interval  of interest.
Since triggering condition in (\ref{estimates_lattest}) is defined based on the triggering threshold $\delta$, we will also analyze the  relationship between $\lambda$ and $\delta$ in the following.

It is noted that since the local interactions of agents are complex and related with nonlinear operation (i.e., the inverse of matrix), it seems impossible to accurately quantify the communication rate of event-triggered distributed filters. In the subsequent, we propose a novel analysis approach to obtain the minimal (maximal) successive non-triggering (triggering) times, which results in a conservative communication rate, i.e., an upper bound of communication rate can be given.	
Suppose $\mathbb{T}_1(i)=[t_{k_1}(i):t_{k_2}(i)]$ is one period during which the event in (\ref{estimates_lattest}) is successively triggered, and $\mathbb{T}_2(i)=[t_{l_1}(i):t_{l_2}(i)]$ is one period during which the event is successively not triggered. Thus, for each agent, the entire time domain consists of subintervals like $\mathbb{T}_1(i)$ and $\mathbb{T}_2(i)$.
In the following two lemmas, we will provide conditions which ensure the maximal interval of $\mathbb{T}_1(i)$ and the minimal interval of $\mathbb{T}_2(i)$.

\begin{lemma}\label{lem_trigger}
	Consider the system (1)-(3) with nonsingular state transition matrix. For EPDKF algorithm with triggering threshold $\delta\geq 0$, if the event of agent $i$ is successively triggered in the interval $[t_0:t_1]$, then
	\begin{align}
\bar f(t,\delta,i)>0, \forall  t\in [0:t_1-t_0], 
	\end{align}
 where $\bar f(t,\delta,i)$ is  defined in (\ref{eq_ferror}).
\end{lemma}
\begin{proof}
		Please see the proof in Appendix \ref{pf_lemma_tri1}.
\end{proof}
\begin{lemma}\label{lem_trigger2}
	Consider the system (1)-(3) with nonsingular state transition matrix. For EPDKF algorithm with triggering threshold $\delta\geq 0$,  the event of agent $i$ is successively not triggered in the interval $[t_2:t_3]$, if 
	\begin{align*}
\bar g(\delta,t,i)\leq 0, \forall t\in [0:t_3-t_2],
	\end{align*}
 where $\bar g(\delta,t,i)$ is an operator defined in (\ref{eq_ferror2}).
\end{lemma}
\begin{proof}
	Please see the proof in Appendix \ref{pf_lemma_tri2}.
\end{proof}
Suppose $[0:T]$ is the time interval of interest. The maximal successive triggering time $T_1(i)$ and the minimal non-triggering time $T_2(i)$   can be obtained by solving  Problems 1 and 2 in the following, respectively.

\emph{Problem 1}: Given $\delta\geq 0$, for	$i\in\mathcal{V}$,
\begin{align}
T_1(i)=\max_{t\in[0:T]} t, \text{ subject to } \bar f(t,\delta,i)>0,
\end{align}
where $\bar f(t,\delta,i)$ is defined in (\ref{eq_ferror}).

\emph{Problem 2}: Given $\delta\geq 0$, for	$i\in\mathcal{V}$,
\begin{align}
T_2(i)=\max_{t\in[0:T]} t,  \text{ subject to } \bar g(\delta,t,i)\leq 0,
\end{align}
where  $\bar g(\delta,t,i)$ is defined in (\ref{eq_ferror2}).

\begin{remark}
	For a decision maker with global system information, the optimization problems 1 and 2 are easily to be solved offline, since the constraints of the problems can be linearly expanded with time $t$.
\end{remark}
Recall that $\beta$ is given in (\ref{P_recursive}) and  the denotation  $|\mathcal{N}_{i,out}^0|$  is the out-degree of agent $i$ without including itself. 
Let $\mathcal{I}_{ \{x\geq y\}}$ be an indicative function of $\{0,1\}$ judging whether $x\geq y$ holds.
Then, we have the following theorem on communication rate. 
\begin{theorem}\label{thm_rate}
		(\textbf{Communication rate}) 
	Consider the system (1)-(3) with nonsingular state transition matrix. For EPDKF algorithm with triggering threshold $\delta\geq 0$, 
   given the time interval of interest  $[0:T]$, 
	if there exists a nonempty set $\mathcal{V}_1\subseteq \mathcal{V}$, such that for $\forall i\in\mathcal{V}_1$, 
	the following conditions hold,\\
1) Problem 1 has a feasible solution $T_1(i)\in [0:T]$, subject to $\mathcal{I}_{\{T_1(i)\geq 2\}}\sum_{\tau=2}^{T_1(i)} \beta^{\tau}(A^{-\tau})^{T}A^{-\tau}\leq I_{n}$.  	\\
2) Problem 2 has a feasible solution $T_2(i)\in [0:T-T_1(i)]$. \\
then the communication rate $\lambda (\delta)$ is no larger than $\lambda_0(\delta)\in(0,1)$, i.e.,
\begin{align}
0<\lambda (\delta)\leq \lambda_0(\delta)<1,
\end{align}
where 
\begin{align}\label{eq_rate0}
\lambda_0(\delta)=1-\frac{\sum_{i\in\mathcal{V}_1}T_{2}(i)\left\lfloor\frac{T}{T_{1}(i)+T_{2}(i)}\right\rfloor |\mathcal{N}_{i,out}^0| }{T\sum_{i\in\mathcal{V}}|\mathcal{N}_{i,out}^0|}.
\end{align}
Furthermore, $\lambda_0(\delta)$ is  monotonic decreasing function of $\delta$.
\end{theorem}
\begin{proof}
Based on the definition of communication rate in (\ref{eq_rate}), if conditions 1) and 2) of this theorem hold,  (\ref{eq_rate0}) can be calculated.  
According to Lemma \ref{lem_trigger}, the  length of real successive triggering times is smaller than that calculated based on  Problem 1. From Lemma \ref{lem_trigger2}, 	we see that the  length of real successive triggering times is larger than that calculated based on  Problem 2. Thus, according to (\ref{eq_rate0}), the communication rate of the network $\lambda (\delta)$ is no larger than $\lambda_0(\delta)\in(0,1)$. 

To prove the monotonicity of $\lambda_0(\delta)$ with respect to $\delta$, next we will analyze the detailed form of $ z_{ t,i}(A,H,R,D,\mathcal{A},\Upsilon, \beta,\delta)$ in (\ref{eq_lower}). According to (\ref{pf_rate}), 
\begin{align*}
&z_{ t,i}(A,H,R,D,\mathcal{A},\Upsilon, \beta,\delta)\\
=&\bar z_{ t,i}(A,H,R,D,\mathcal{A},\Upsilon, \beta)-\delta\mathcal{I}_{t\geq 2}\sum_{\tau=2}^{t} \beta^{\tau}(A^{-\tau})^{T}A^{-\tau}
\end{align*}
where $\mathcal{I}_{ \{x\geq y\}}$ be an indicative function of $\{0,1\}$ judging whether $x\geq y$ holds, and $\bar z_{ t,i}(A,H,R,D,\mathcal{A},\Upsilon, \beta)$ is  a constant matrix derived based on predefined system matrices or vectors $A,H,R,D,\mathcal{A},\Upsilon, \beta$, and  the predefined integer $t$. According to (\ref{eq_ferror}), $\bar f(t,\delta,i)>0$ is equivalent to $\lambda_{max}\{J(\cdot)\}>0$, where
\begin{align*}
J(\cdot)&= f_{t,i}(\cdot)-eig_{pos}\left(z_{t,i}(\cdot)+\delta I_n\right),
\end{align*}
where $eig_{pos}(\cdot)$ is the matrix operator defined in (\ref{eq_ferror}), and 
\begin{align*}
&z_{t,i}(\cdot)+\delta I_n=\bar z_{t,i}(\cdot)+\delta\left(I_{n}-\mathcal{I}_{t\geq 2}\sum_{\tau=2}^{t} \beta^{\tau}(A^{-\tau})^{T}A^{-\tau}\right)
\end{align*}
Due to the condition  $\mathcal{I}_{\{T_1(i)\geq 2\}}\sum_{\tau=2}^{T_1(i)} \beta^{\tau}(A^{-\tau})^{T}A^{-\tau}\leq I_{n}$ in 1),
we see if $\delta$ is bigger, the length of successive triggering times (i.e., $t_1$) is smaller. Besides, considering (\ref{eq_ferror2}), if $\delta$ is bigger, the length of successive non-triggering times (i.e., $t_2$) is larger. Thus, the calculated $\lambda_0(\delta)$ is a monotonic decreasing function of $\delta.$
\end{proof}
	
If the  time interval of interest is very large (e.g., $T\rightarrow\infty$), we have the following corollary.
\begin{corollary}\label{coro_rate}
	Under the same conditions as Theorem \ref{thm_rate}, if there exists an integer $T_b$, such that for any $T\geq T_b$, $T_1(i)+T_2(i)\leq T_b$, then 
	\begin{align}
	0<\lambda (\delta)\leq \bar\lambda_0(\delta)=\lim\limits_{T\rightarrow\infty}\lambda_0(\delta)<1,
	\end{align}
	where 
	\begin{align}\label{eq_rate02}
	\bar \lambda_0(\delta)=1-\frac{\sum_{i\in\mathcal{V}_1}\frac{T_{2}(i)}{T_{1}(i)+T_{2}(i)}|\mathcal{N}_{i,out}^0|}{\sum_{i\in\mathcal{V}}|\mathcal{N}_{i,out}^0|}.
	\end{align}
\end{corollary}	
\begin{proof}
If there exists an integer $T_b$, such that for any $T\geq T_b$, $T_1(i)+T_2(i)\leq T_b$, then $\lim\limits_{T\rightarrow\infty}\frac{1}{T}\left\lfloor\frac{T}{T_{1}(i)+T_{2}(i)}\right\rfloor=\frac{1}{T_{1}(i)+T_{2}(i)}$. According to (\ref{eq_rate0}), the conclusion holds.
\end{proof}
\begin{remark}
	Theorem \ref{thm_rate} and Corollary \ref{coro_rate} study the communication rate for the distributed event-triggered filter and provide  methods to obtain the minimal (maximal) successive non-triggering (triggering) times, which seem to be not investigated in the existing results \cite{liu2015eventinformatics,liu2015event,Battistelli2016Distributed} of distributed filtering to our knowledge.
\end{remark}	
	\section{Numerical simulation}
	In this section, we consider a state constraint navigation problem of a land-based vehicle, which was widely studied \cite{simon2010kalman,simon2002kalman}, in order to illustrate the effectiveness of the proposed TPDKF and EPDKF. 
	Corresponding to the system (\ref{system_all})-(\ref{system_all1}), 
	the  vehicle dynamics can be approximated through setting 
	$A_{k}=\left[\begin{smallmatrix}
	1 & 0 & T_s& 0\\  
	0 &1 & 0 & T_s\\  
	0 & 0 & 1& 0\\  
	0 &0 & 0 & 1\\  
	\end{smallmatrix}	\right],$
	and the first two state elements of $x_{k}$ are the north and east positions, and the last two are the north and east velocities,
	$\omega_k$ is the process noise whose covariance matrix is upper bounded by $Q=diag\{4,4,1,1\}$; and $v_{k,i}$ is the measurement noise of $i$th agent with the covariance matrix $R_i=90$.
	The time interval of measurements is $[0,25]$ with the sampling period $T_s$ is 0.1s, which means $k\in[0:250]$, and
	the initial state is generated by a Gaussian distribution with zero mean and covariance  $[100,100,\tan^2\theta,1]^T$. 
	If one agent has the knowledge of the vehicle running on a road with the heading of $\theta$ equal to $60$ degrees, then
	$\tan\theta=x_{k}(1)/x_{k}(2)=x_{k}(3)/x_{k}(4)$. The constraint can be rewritten in the form of $D_kx_{k}=0$ with $	D_k=\left(
	\begin{smallmatrix}
	1 & -\tan\theta& 0 & 0\\
	0  &  0       &  1&-\tan\theta
	\end{smallmatrix}
	\right),k=1,2,\ldots.$
	In the following, the elements of the weighted adjacent matrix $\mathcal{A}$ are given by the so-called Metropolis weights (\cite{Battistelli2014Kullback}). The initial settings of the algorithms are
	$\hat x_{0,i}=[0,0,0,0]^T,P_{0,i}=diag\{100,100,4,4\}, \varepsilon_i=0.01,i=1,\ldots,N.$
	To better show the estimation performance of the proposed TPDKF, in the following, we compare the proposed TPDKF with the centralized  Kalman filter (CKF) and Distributed State Estimation with Consensus on the Posteriors (DSEA-CP) \cite{Battistelli2014Kullback}.
	CKF is the  minimum-variance centralized filter for linear dynamic systems and
	the algorithm DSEA-CP considers the distributed filter based on consensus. 
	We conduct the numerical simulation through Monte Carlo experiment, in which $1000$ Monte Carlo trials for TPDKF, CKF and DSEA-CP are performed, respectively. The mean square error, averaged over all the agents, is defined as
$	MSE_{k}=\frac{1}{N}\sum_{i=1}^{N}\left[ \frac{1}{500}\sum_{j=1}^{500}(\hat x_{k,i}^j-x_{k}^j)^T(\hat x_{k,i}^j-x_{k}^j)\right],$
	{where $\hat x_{k,i}^j$ is the state estimate of the $j$th trail of agent $i$ at time $k$.}
	In addition, considering the proposed TPDKF, the upper bound of MSE defined above is given as $tr(P_{k})=\frac{1}{N} \sum_{i=1}^N tr(P_{k,i})$.
	
	\subsection{Performance Evaluation: Case 1}
	
	In this subsection, we focus on the extended collective observability condition (\ref{ass_A2}) for the performance of the proposed algorithms by only considering three agents in the network shown in Fig. \ref{topology3}.    The state constraints of the agents are assumed to be
	$D_{k,1}=D_{k,3}=\left(
	\begin{smallmatrix}
	1 & -\tan\theta& 0 & 0\\
	0  &  0       &  1&-\tan\theta
	\end{smallmatrix}
	\right),
	D_{k,2}=\left(
	\begin{smallmatrix}
	0 & 0& 0 & 0
	\end{smallmatrix}
	\right),d_{k,i}=0,i=1,2,3.$
	The observation matrices are supposed to have the following forms $	H_{k,1}=H_{k,3}=\left(
	\begin{smallmatrix}
	1 & 0& 0 & 0
	\end{smallmatrix}
	\right),
	H_{k,2}=\left(
	\begin{smallmatrix}
	0 & 0& 0 & 0
	\end{smallmatrix}
	\right).$
	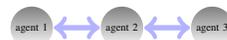
\begin{figure}[htp]
		\begin{center}
			\begin{tikzpicture}[scale=0.4, transform shape,line width=2pt]
			\tikzstyle{every node} = [circle,shade, fill=gray!30]
			\node (a) at (0, 0) {agent 1};
			\node (b) at +(0: 1.5*2) {agent 2};
			\node (c) at +(0: 1.5*4) {agent 3};
			\foreach \from/\to in {a/b, b/c}
			\draw [blue!30,<->] (\from) -- (\to);
			\end{tikzpicture}
		\end{center}
		\caption{The communication illustration of a agent network with 3 agents}\label{topology3}
	\end{figure}
	The MSE comparison of TPDKF, CKF and DSEA-CP  is given in Fig. \ref{case12}. It can be seen that the algorithms CKF and DSEA-CP are both divergent, but the MSE of TPDKF and the corresponding upper bounds (i.e., $\sum_{i=1}^{N}tr(P_{k,i})$) still remain stable if $L\geq 1$. Since  the collective observability condition  is not satisfied, the CKF is divergent. Yet the extended collective observability in Assumption \ref{ass_observable} holds due to the contribution of state constraints. Besides, as $L$ increases, TPDKF achieves better  estimation performance in terms of MSE and tr(P). 
	The results demonstrate that the state constraints can relax the observability condition of the algorithms,  TPDKF  can efficiently employ the information of state constraints and the fusion-projection number $L$ has direct influence on estimation performance of TPDKF.

		\begin{figure}[htp]
			\centering
			\includegraphics[width=0.3\textwidth]{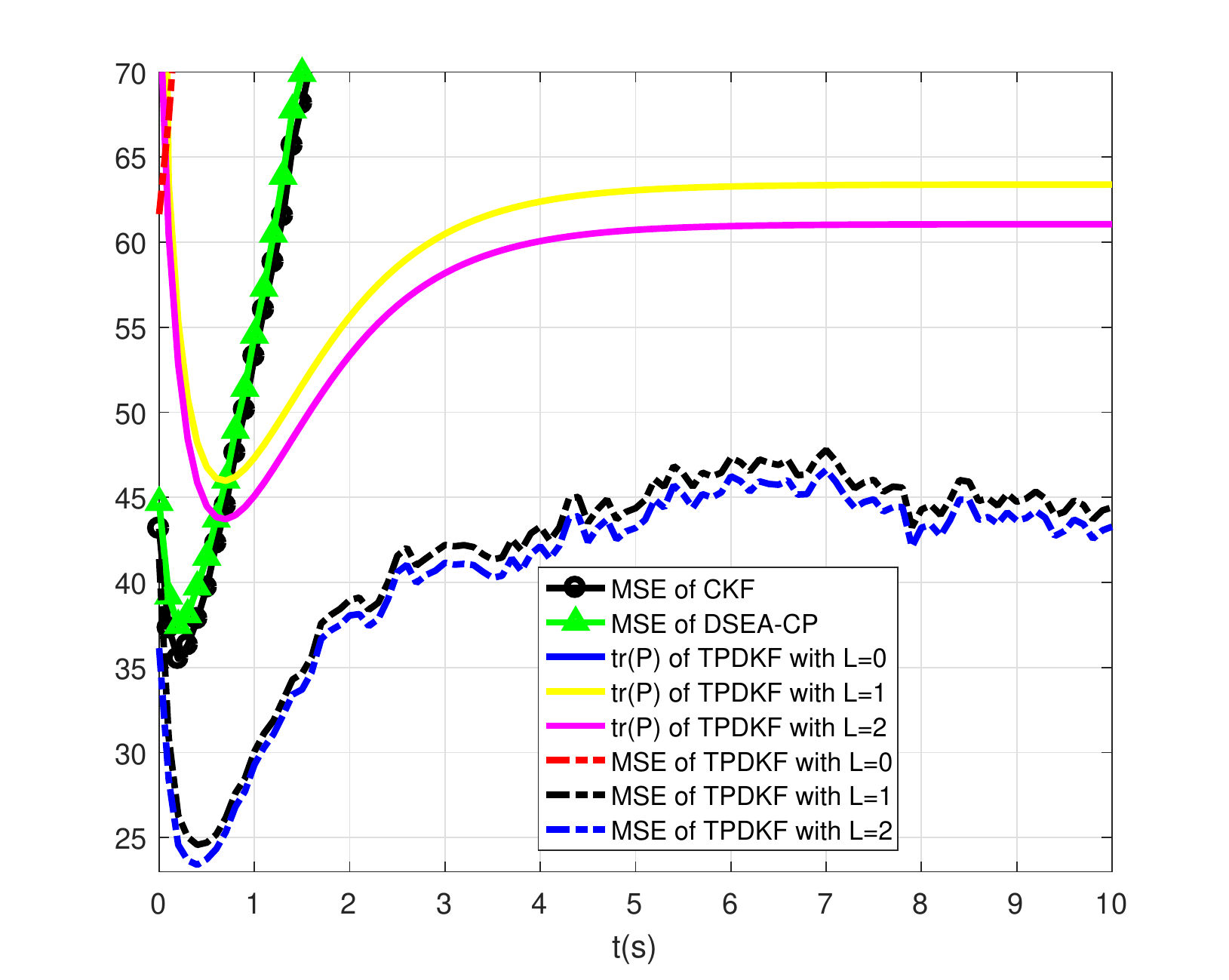}
			\caption{Performance comparision among CKF, DSEA-CP and TPDKF}
			\label{case12}
		\end{figure}

	To show the effectiveness of the proposed EPDKF in Table \ref{ODKF3}, we consider the simulation performance under the initial settings.
	The event triggering thresholds are set to  $\delta_{1}=0.3,\delta_{2}=0.4,\delta_{3}=0.8$. By simulation, the communication rate is $\lambda=0.311$, which effectively reduces the communication burden of the network. Fig. \ref{triggering_instants} shows the triggering instants of each agent. 
	  {Fig. \ref{single_consitence} reveals that the estimation error of EPDKF remains stable and the consistency  holds.
	Fig. \ref{triggering_instants} illustrates that there is a periodic transmission behavior of each agent. As a result, Fig. \ref{single_consitence} shows some  periodicity of the estimation performance.}
	{
	Table \ref{table_compare} shows the relationships between  $\delta$, $tr( MSE_e)=\frac{1}{3} \sum_{i=1}^3 \max_{k\geq 50}(tr(MSE_{k,i}))$ and $tr( P_e)=\frac{1}{3} \sum_{i=1}^3 \max_{k\geq 50}(tr(P_{k,i}))$.
	Fig. \ref{trace_P_lambda} reveals the dynamic changing of estimation performance along with the communication rate $\lambda$. }It can be seen that the event triggering thresholds decrease, and meanwhile, the estimation error decreases, as the communication rate increases.
	\begin{figure}[htp]
		\centering
		\includegraphics[width=0.3\textwidth]{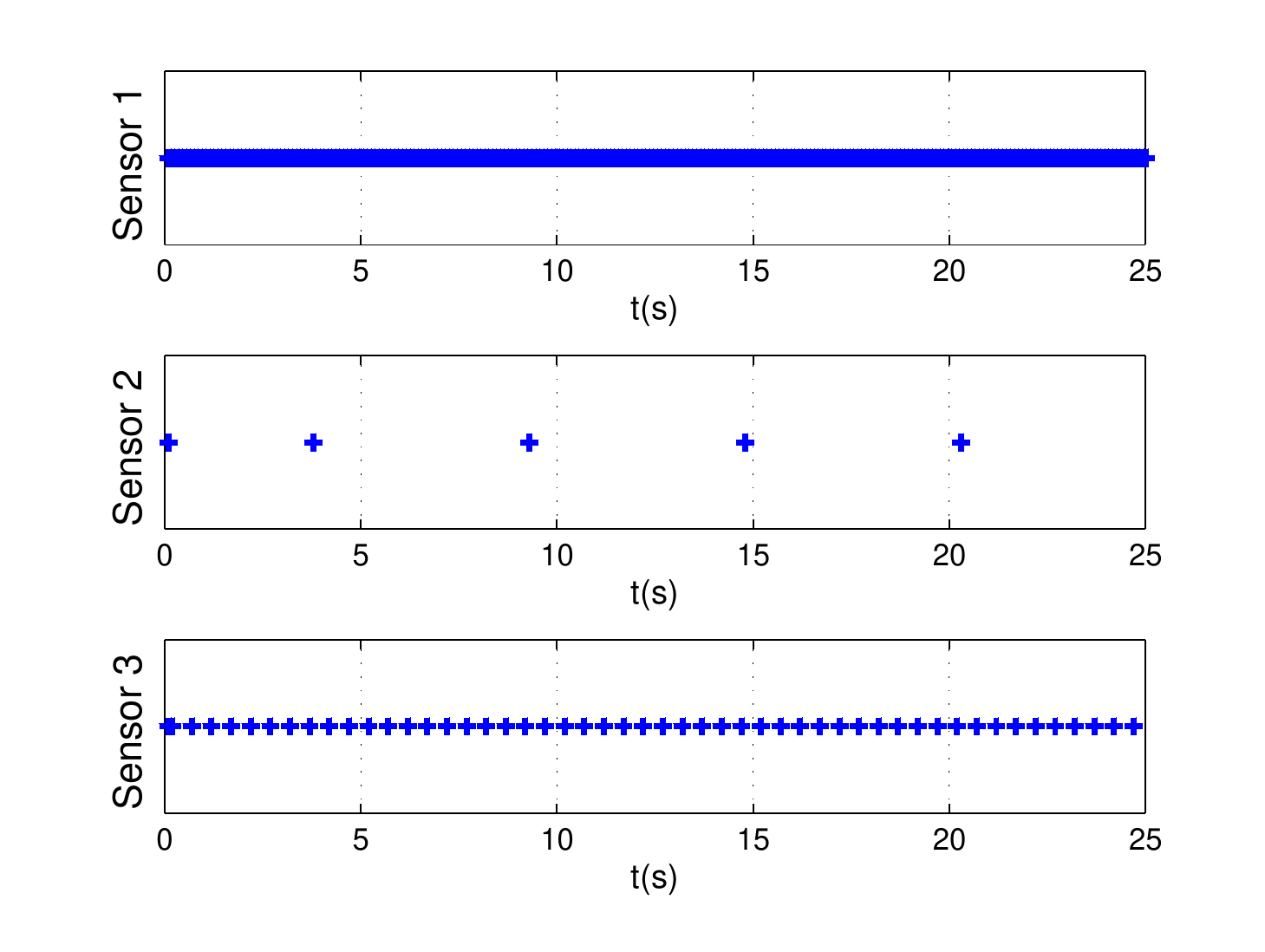}
		\caption{Triggering instants of agents}\label{triggering_instants}
	\end{figure}
	\begin{figure}[htp]
		\centering
		\includegraphics[width=0.3\textwidth]{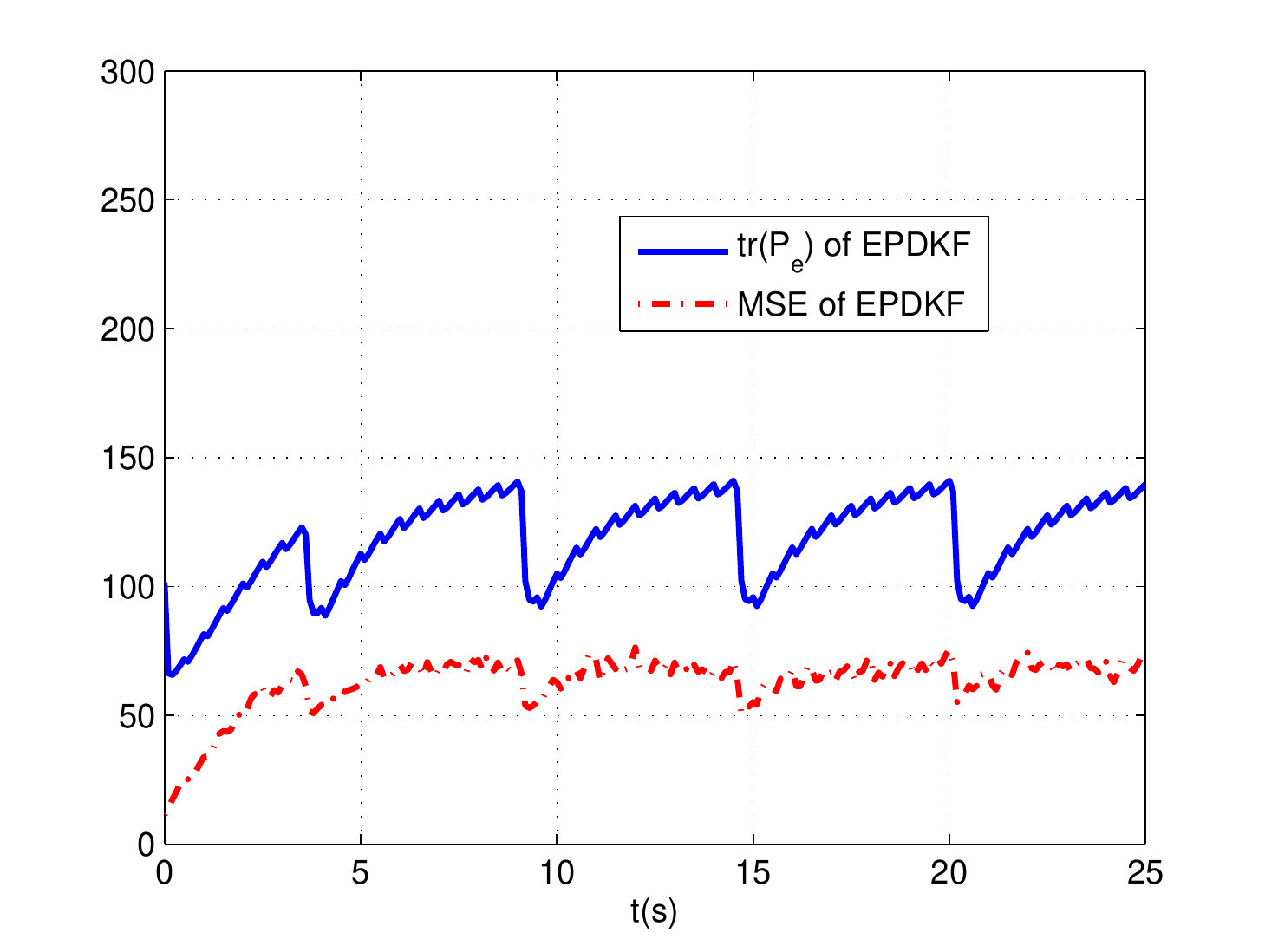}
		\caption{The estimation performance of EPDKF}\label{single_consitence}
	\end{figure}

	\begin{figure}[htp]
		\centering
		\includegraphics[width=0.3\textwidth]{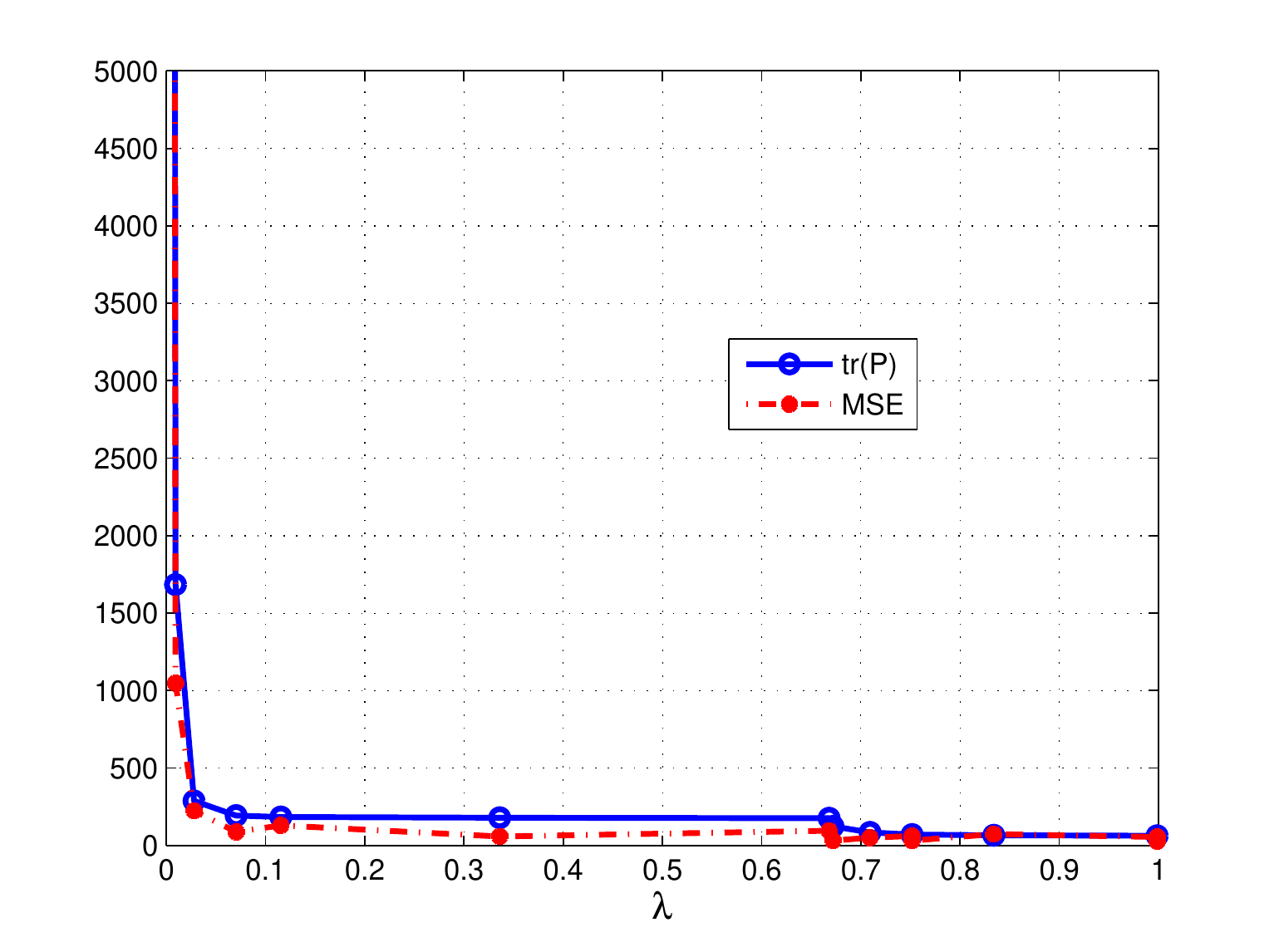}
		\caption{The relationships among tr(P), MSE and communication rate $\lambda$}\label{trace_P_lambda}
	\end{figure}

	\begin{table}[htp]
		\centering
		\caption{The relationship of $tr(P_e)$, $MSE_e$,   $\delta$ and $\lambda$}\label{table_compare}
		{
			\begin{tabular}{lccccc}  
				\hline
				$\quad\quad\delta_{i}=\delta\quad\quad\quad\quad\quad$ 2.00&0.97&0.57&0.42& 0.12&\\ \hline  
				trace($P_e$)$\quad\quad\qquad\quad$ 3.76e4&3.77e3&179.19&141.99&101.53 &   \\       
				\hline  
				trace($MSE_e$)$\qquad\quad$ 2.01e4&2.77e3&124.40&78.35&49.71 &   \\  \hline       
			\end{tabular}
		}
	\end{table}

	\subsection{Performance Evaluation: Case 2}
	
	In this subsection, we consider 20 agents for illustration of the overall performance, where the observation matrices of these agents are uniformly randomly selected from 	$H_{k,1}=\left(
	\begin{array}{cccc}
	1 & 0& 0 & 0
	\end{array}
	\right),H_{k,2}=\left(
	\begin{array}{cccc}
	0 & 0.3& 0 & 0
	\end{array}
	\right),H_{k,3}=\left(
	\begin{array}{cccc}
	0 & 1& 0 & 0
	\end{array}
	\right)$. 
	The network communication topology is illustrated in Fig. \ref{case221}. It is noted that for the system, regarding all the observation matrices listed below, the traditional collective observability condition is satisfied. The state constraints of the agents are uniformly randomly selected from 
	$	D_{k,1}=\left(
	\begin{smallmatrix}
	1 & -\tan\theta& 0 & 0\\
	0  &  0       &  1&-\tan\theta
	\end{smallmatrix}
	\right),
	D_{k,2}=\left(
	\begin{smallmatrix}
	0 & 0& 0 & 0
	\end{smallmatrix}
	\right),d_{k,i}=0,i=1,2$.
	\begin{figure}[htp]
		\centering
		\includegraphics[width=0.3\textwidth]{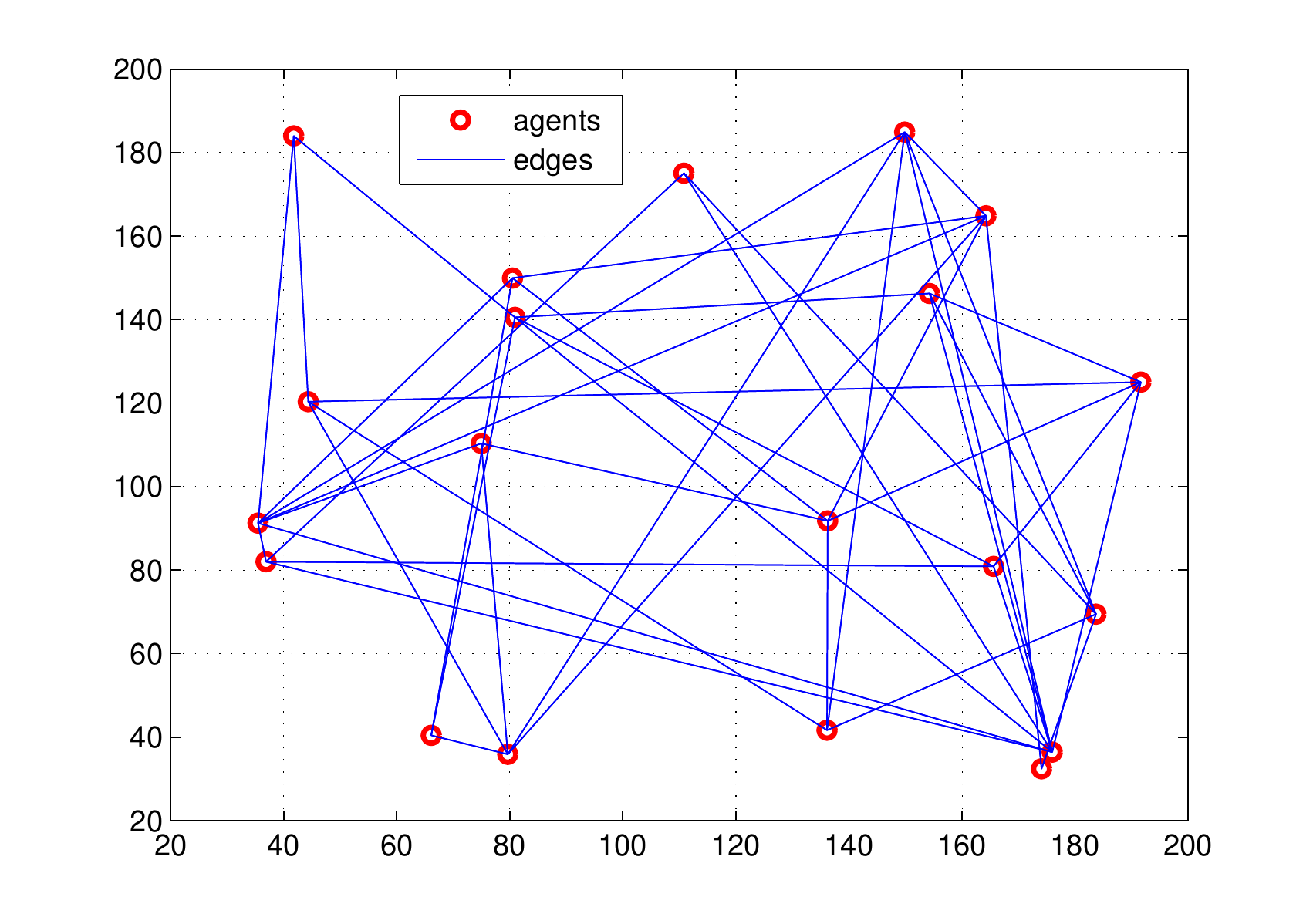}
		\caption{Communication topology of a network with 20 agents}
		\label{case221}
	\end{figure}
	\begin{figure}[htp]
		\centering
		\includegraphics[width=0.3\textwidth]{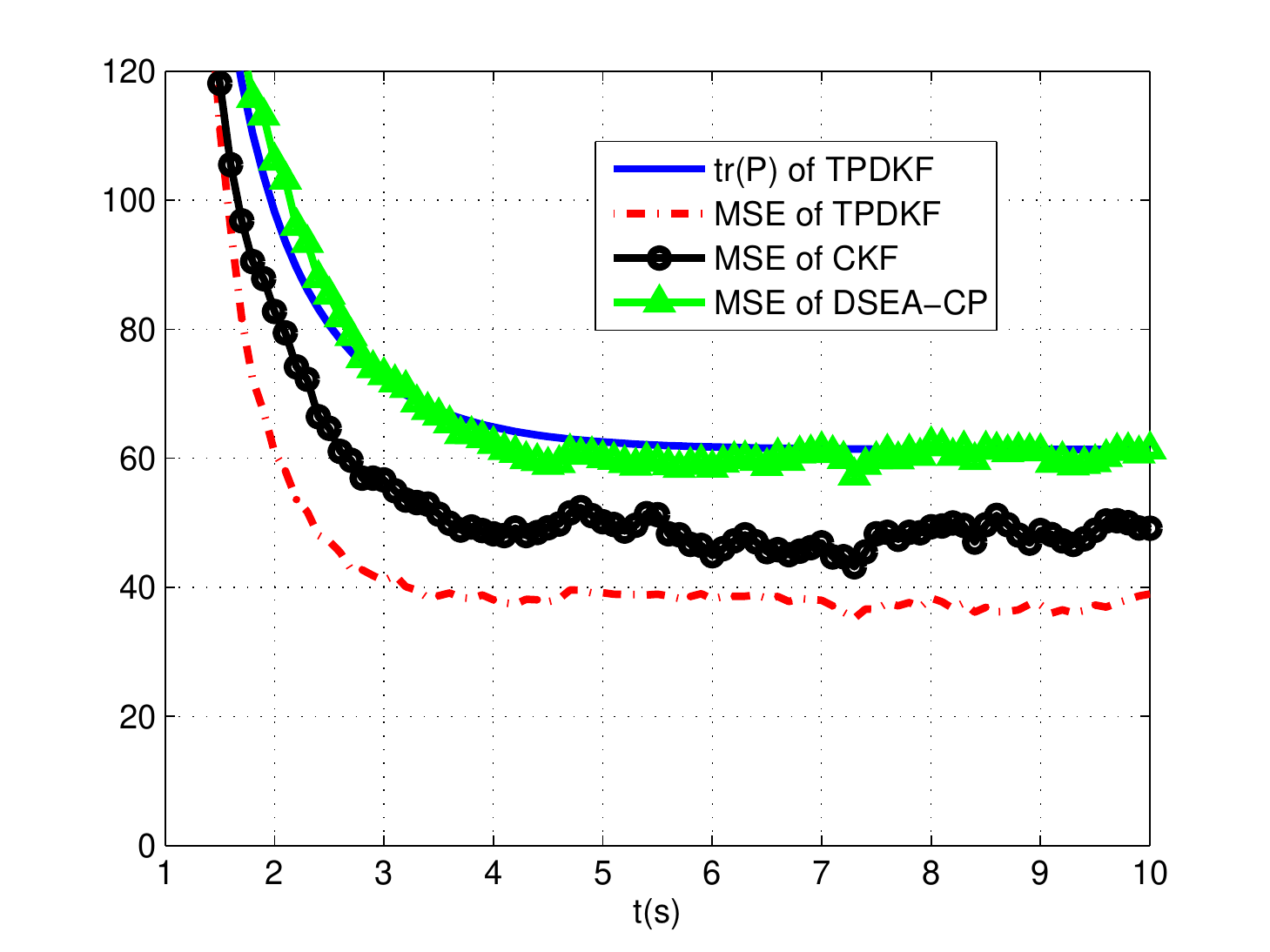}
		\caption{Performance comparison among TPDKF, CKF and DSEA-CP}
		\label{case211}
	\end{figure}	
			\begin{figure}[htp]
				\centering
				\includegraphics[width=0.3\textwidth]{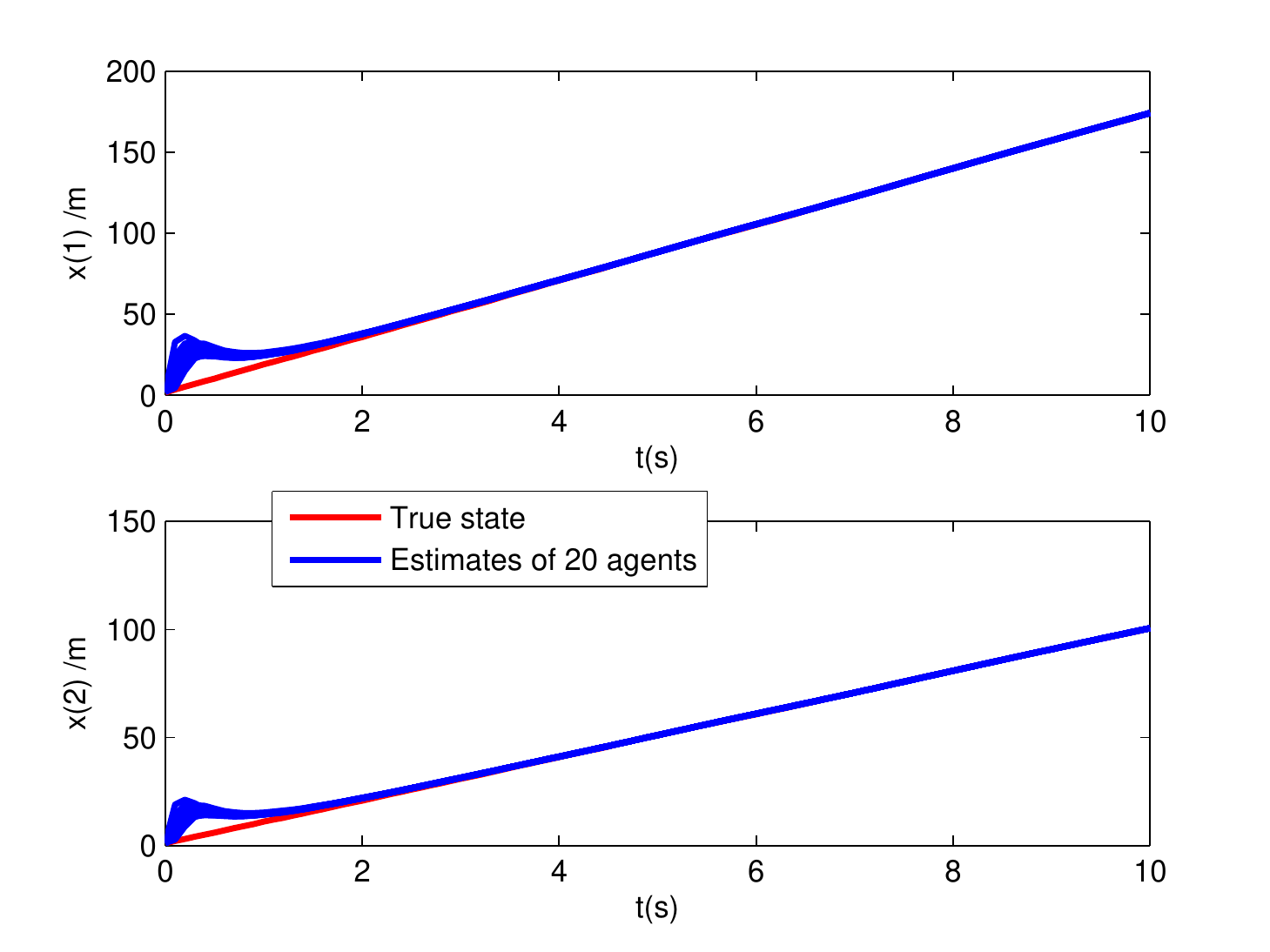}
				\caption{The trajectories of 20 agents}\label{states_agents}
			\end{figure}	
	In this scenario, the MSE comparison of TPDKF with $L=1$, CKF and DSEA-CP is given in Fig. \ref{case211}.
	According to Fig. \ref{case211}, we see that the  consistency of  TPDKF in Lemma \ref{thm_pro} is illustrated.
	Additionally, although the three algorithms are all stable, the proposed TPDKF has better estimation performance than
	DSEA-CP and even than CKF, which effectively shows that the state constraints in our algorithm can improve the estimation precision.
{	From Fig. \ref{states_agents}, it can be seen that the state estimates of 20 agents using the proposed TPDKF with different initial values can well track the stochastic dynamics of the system. }
	\section{Conclusions}
	In this paper, we investigated the problem of distributed estimation with SEC under time-based and event-triggered communication schemes, respectively.
	We proposed a time-driven distributed Kalman filter by combining a filtering structure and a fusion-projection
		operator. Then we provided several essential properties of the filter under some mild conditions of the system.
	Furthermore, we analyzed the influence of fusion-projection step to the filtering performance. 	
	Moreover, we proposed a   distributed event-triggered filter with SEC and  provided the design principle of the triggering thresholds.
More importantly, we analyzed the communication rate for the distributed event-triggered filter and  provided  methods to obtain the minimal (maximal) successive non-triggering (triggering) times. 
\section{Acknowledgment}
	The authors would like to thank the Associate Editor and
	anonymous reviewers for their very constructive comments,
	which greatly improved the quality of this work.
	
	\ifCLASSOPTIONcaptionsoff
	\newpage
	\fi

	\appendices
	\section{Proof of Lemma \ref{lem_singular}}\label{pf_lem_singular}
	Due to the form of $\hat P_{k,i}^{l+1}$ and $\epsilon_i>0$, it is easy to obtain $\hat P_{k,i}^{l+1}\leq \tilde P_{k,i}^{l+1}.$ 		
	Employing  matrix inverse formula yields 
	\begin{align}\label{eq_inverse2}
	(\tilde P_{k,i}^{l+1})^{-1}=(\check P_{k,i}^l)^{-1}+\frac{1}{\varepsilon_i}D_{k,i}^TD_{k,i}>0, \forall \varepsilon_i>0,
	\end{align}
	which means all  eigenvalues of $\tilde P_{k,i}^{l+1}$ are positive. Thus, $\tilde P_{k,i}^{l+1}$  is a positive definite matrix.	
	Also, it is straightforward to see $\tilde P_{k,i}^{l+1}\leq \check P_{k,i}^l, l=0,1,\cdots,L-1$.
	To prove the second conclusion, 
 we employ a matrix approximation method.
	Given any scalar $\eta>0$, we denote $\hat P_{k,i}^{l+1}(\eta)=\check P_{k,i}^l-\check P_{k,i}^lD_{k,i}^T(D_{k,i}\check P_{k,i}^lD_{k,i}^T+\eta I_{s_i})^{-1}D_{k,i}\check P_{k,i}^l.$
	It can be seen that $\hat P_{k,i}^{l+1}\leq\hat P_{k,i}^{l+1}(\eta)$. 
	Employing  matrix inverse formula yields $(\hat P_{k,i}^{l+1}(\eta))^{-1}=(\check P_{k,i}^l)^{-1}+\frac{1}{\eta}D_{k,i}^TD_{k,i}.$
	Since $D_{k,i}$ is of rank  $s_{k,i}$,  $D_{k,i}^TD_{k,i}$ is of rank   $s_{k,i}$. Let $\eta\longrightarrow 0^+$, then we have $\hat P_{k,i}^{l+1}\leq\lim\limits_{\eta\longrightarrow 0^+ }\hat P_{k,i}^{l+1}(\eta):=S_*$.
	It can be seen  that $S_*$ has $d_{k,i}$ eigenvalues of zero. Then, the conclusion of this lemma holds.	
	\section{Proof of Proposition \ref{thm_distri}}\label{pf_pro1}
	{
	According to TPDKF in Table \ref{ODKF2}, let $l\in[0:L-1]$, then 
	$\tilde e_{k,i}^{l+1}=\left[ I_n-\check P_{k,i}^lD_{k,i}^T(D_{k,i}\check P_{k,i}^lD_{k,i}^T)^{-}D_{k,i}\right] \check  e_{k,i}^l,$
	where the estimation  error of the local fusion  $\check e_{k,i}^l$ satisfies
	$\check e_{k,i}^l \check P_{k,i}^l\sum_{j\in \mathcal{N}_{i}}a_{i,j}(\tilde P_{k,j}^l)^{-1}\tilde e_{k,j}^l.$	
	The error of the measurement update $\tilde e_{k,i}=\tilde e_{k,i}^0$ can be derived through
	$\tilde e_{k,i}=(I-K_{k,i}H_{k,i})\bar e_{k,i}+K_{k,i}v_{k,i}.$
	Additionally, the prediction error follows from $\bar e_{k,i}=A_{k-1}e_{k-1,i}^L+\omega_{k-1}.$
	Since $\omega_{k-1},v_{k,i}$ and $e_{0,i}=\hat x_{0,i}-x_{0}$ are Gaussian, $e_{k,i}$ is also Gaussian by employing the inductive method and the property of Gaussian distribution.	
}

	\section{Proof of Lemma \ref{thm_pro}}\label{pf_th1}
	 Since the state estimate $\tilde x_{k,i}^{l+1}$ is a solution of the optimization problem with state equality constraints, it naturally  satisfies the SEC of agent $i$ at time $k$. Next, we use the inductive method to show the consistency of the pairs $(\tilde x_{k,i}^{l+1},P_{k,i}^{l+1})$ for $l\leq L-1$. 
		At the initial moment, given $\theta_i>0$,
	\begin{align*}
	&E\{(\hat x_{0,i}-x_{0})(\hat x_{0,i}-x_{0})^T\}\\
	\leq&\frac{1+\theta_i}{\theta_i}(\hat x_{0,i}-E\{x_{0})(\hat x_{0,i}-E\{x_{0}\})^T\\
	&+(1+\theta_i)E\{(x_{0}-E\{x_{0}\})(x_{0}-E\{x_{0}\})^T\}\\
	\leq&\frac{\theta_i+1}{\theta_i}P_{0,i}^*+(1+\theta_i)P_0\leq P_{0,i},
	\end{align*}
	where $(\hat x_{0,i}-E\{x_0\})(\hat x_{0,i}-E\{x_0\})^T\leq P_{0,i}^*$ and $E\{(x_{0}-E\{x_{0}\})(x_{0}-E\{x_{0}\})^T\}\leq P_{0}$.	
	Let us assume that, at  time $k-1$, $E\{(\hat x_{k-1,i}-x_{k-1})(\hat x_{k-1,i}-x_{k-1})^T\} \leq P_{k-1,i}.$
	For the prediction estimation error at time $k$,  $\bar e_{k,i}=A_{k-1}e_{k-1,i}+\omega_{k-1}.$
	As $E\{\omega_{k-1}\omega_{k-1}^T\}\leq Q_{k-1}$, one can obtain $E\{\bar e_{k,i}\bar e_{k,i}^T\}\leq  A_{k-1}E\{e_{k-1,i}e_{k-1,i}^T\}A_{k-1}^T+Q_{k-1}\leq \bar P_{k,i}$.
	Recall $\tilde e_{k,i}=(I-K_{k,i}H_{k,i})\bar e_{k,i}+K_{k,i}v_{k,i}.$
	Since $E\{v_{k,i}v_{k,i}^T\}\leq R_{k,i}$, it then follows that $E\{\tilde e_{k,i}(\tilde e_{k,i})^T\}\leq \tilde P_{k,i}$. 
	Due to the consistent estimation of CI strategy (\cite{Julier1997A}), one can obtain that $	E\{\check e_{k,i}^l(\check e_{k,i}^l)^T\}\leq \check P_{k,i}^l, l=0,1,\cdots,L-1.$	
	Note the estimation error satisfies $\tilde e_{k,i}^{l+1}=P_{k,i}^{D,l}\check  e_{k,i}^{l}$,
	where $P_{k,i}^{D,l}= I_n-\check P_{k,i}^lD_{k,i}^T(D_{k,i}\check P_{k,i}^lD_{k,i}^T)^{-}D_{k,i}$. Due to $E\{(\check e_{k,i}^l)(\check e_{k,i}^l)^T\}\leq \check P_{k,i}^l$ and the Moore-Penrose inverse property $A^-AA^-=A^-$,
we have
	\begin{equation*}\label{pf2_12}
	\begin{split}
	&	E\{\check e_{k,i}^{l+1}(\check e_{k,i}^{l+1})^T\}\\
	=&P_{k,i}^{D,l}E\{\check e_{k,i}^l(\check e_{k,i}^l)^T\}(P_{k,i}^{D,l} )^T,\\
	\leq &P_{k,i}^{D,l} \check P_{k,i}^l(P_{k,i}^{D,l} )^T,\\
	=&\check P_{k,i}^l-\check P_{k,i}^lD_{k,i}^T(D_{k,i}\check P_{k,i}^lD_{k,i}^T)^{-}D_{k,i}\check P_{k,i}^l\leq \tilde P_{k,i}^{l+1}.
	\end{split}
	\end{equation*}		
	Then $E\{e_{k,i}e_{k,i}^T\}=E\{\check e_{k,i}^{L}(\check e_{k,i}^{L})^T\} \leq \tilde P_{k,i}^{L}=P_{k,i}.$	
	\section{Proof of Lemma \ref{thm_pro2}}\label{pf_thm4}
	{
The proof  of Lemma \ref{thm_pro2} has similar steps as the proof of Lemma \ref{thm_pro} and the proof of Lemma \ref{lem_singular}, hence we only consider the fusion part for convenience.
	Considering the event triggered scheme, agent $i$ can obtain the prediction error of agent $j,j\in\mathcal{N}_i$, based on the received latest $\tilde e_{t,j}$, then  $\tilde e_{k,j}^t=A^{k-t}\tilde e_{t,j}-\sum_{l=t}^{k-1}A^{k-1-l}w_{l},t\leq k.$
	Then considering $E\{\tilde e_{t,j}w_{l}^T\}=0,l\geq t$, 
	\begin{align*}
	&E\{\tilde e_{k,j}^t(\tilde e_{k,j}^t)^T\}\nonumber\\
	\leq &A^{k-t}E\{\tilde e_{t,j}\tilde e_{t,j}^T\}(A^{k-t})^T+\sum_{l=t}^{k-1}A^{k-1-l}Q(A^{k-1-l})^T\nonumber\\\
	\leq& A^{k-t}\tilde P_{t,j}(A^{k-t})^T+\sum_{l=t}^{k-1}A^{k-1-l}Q(A^{k-1-l})^T=\check P_{k,j}^t.
	\end{align*}

	For the proposed EPDKF in Table \ref{ODKF3}, 	
	\begin{equation*}
	\begin{split}
	\check e_{k,i}=\check P_{k,i}\bigg(a_{i,i}\tilde P_{k,i}^{-1}\tilde e_{k,i}+\sum_{j\in \mathcal{N}_{i}^0}a_{i,j} (\tilde P_{k,j}^t)^{-1}\tilde e_{k,j}^t\bigg).
	\end{split}
	\end{equation*}	
According to the consistent estimation of CI strategy (\cite{Julier1997A}), we have $	E\{\check e_{k,i}\check e_{k,i}^T\}\leq P_{k,i}.$
}
\section{Proof of Lemma \ref{lem_trigger}}\label{pf_lemma_tri1}
	If the event in (\ref{estimates_lattest}) is always triggered in the interval $[t_0:t_1]$, then
	\begin{align}\label{eq_tri}
	\lambda_{max}\big(\tilde P_{k,i}^{-1}-\bar{\tilde{P}}_{k,i}^{-1}\big)-\delta>0, k\in [t_0:t_1].
	\end{align}
	Next, we will find a necessary condition of (\ref{eq_tri}). Specifically, we aim to construct an uniformly upper bound of $\tilde P_{k,i}^{-1}$ and an uniformly lower bound of $\bar{\tilde{P}}_{k,i}^{-1}$. Then, we construct a necessary condition of (\ref{eq_tri}) based on the bounds.
	
	Since the time-based algorithm in Table \ref{ODKF2}	with consensus step one has successively triggering times,
	it is easy to conclude that the $\tilde P_{k,i}^{-1}$ matrix of the time-based algorithm is no smaller than that of the event-triggered algorithm, then
	we will analyze the upper bound of $\tilde P_{k,i}^{-1}$ matrix for the time-based algorithm in the following.
	For $\tilde P_{k,i}^{-1}$, $k\in [t_0:t_1]$,
	according to Lemma 1 in \cite{Battistelli2014Kullback}, there exists a scalar $ \bar \beta \in(0,1) $, such that 
	\begin{align}\label{proof_rate1}
	\tilde P_{k,i}^{-1}&=\bar P_{k,i}^{-1} +H_{i}^TR_{i}^{-1}H_{i}\nonumber\\
	&\leq\bar \beta A^{-T}P_{k-1,i}^{-1}A^{-1}+H_{i}^TR_{i}^{-1}H_{i}\nonumber\\
	&=\bar\beta A^{-T}\sum_{j\in \mathcal{N}_{i}}a_{i,j}\tilde P_{k-1,j}^{-1}A^{-1}+H_{i}^TR_{i}^{-1}H_{i}\\
	&\quad+\frac{\bar\beta}{\varepsilon_i } A^{-T}D_{i}^TD_{i}A^{-1}\nonumber.
	\end{align}
	Note $\bar P_{s,j}^{-1}\leq Q^{-1},\forall s\in\mathbb{N},j\in\mathcal{V}$.
	Employing (\ref{proof_rate1}) for $t\in [0:k-t_0]$ times  yields 
	\begin{align}\label{eq_larger}
	\tilde P_{k,i}^{-1}&\leq 
	f_{t,i} (A,H,R,D,\mathcal{A},\Upsilon,Q,\bar\beta),
	\end{align}
	where $f_{t,i} (A,H,R,D,\mathcal{A},\Upsilon,Q,\bar\beta), i\in\mathcal{V}$,  is a constant matrix derived based on predefined system matrices or vectors $A,H,R,D,\mathcal{A},\Upsilon,Q,\bar\beta$, and  the predefined integer $t$.
	
	Next, for the event-triggered algorithm with predefined $\delta$, we will find the uniformly lower bound of $\bar{\tilde{P}}_{k,i}^{-1}$, $k\in [t_0:t_1]$. In light of (\ref{proof_stability213}), 
	\begin{align}\label{pf_rate}
	\bar{\tilde{P}}_{k,i}^{-1}&=(A\tilde P_{k-1,i}A^T+Q)^{-1}\\
	&\geq  \beta A^{-T}(\bar P_{k-1,i}^{-1}+H_{i}^TR_{i}^{-1}H_{i})A^{-1}\nonumber\\
	&\geq  \beta^2 (A^{-2})^{T} P_{k-2,i}^{-1}A^{-2}+ \beta A^{-T}H_{i}^TR_{i}^{-1}H_{i}A^{-1}\nonumber\\
	&\geq \beta^2 (A^{-2})^{T} \sum_{j\in \mathcal{N}_{i}}a_{i,j}(\bar P_{k-2,j}^{-1}+H_{j}^TR_{j}^{-1}H_{j}-\delta I_n)A^{-2}\nonumber\\
	&\quad+ \beta A^{-T}H_{i}^TR_{i}^{-1}H_{i}A^{-1}+ \frac{\beta^2}{\varepsilon_i }(A^{-2})^{T}D_{i}^TD_{i}A^{-2}\nonumber.
	\end{align}
	Noting $\tilde P_{t_0}^{-1}\geq 0$, 
	deriving (\ref{pf_rate}) for $t\in [0:k-t_0]$  times yields 
	\begin{align}\label{eq_lower}
	\bar{\tilde{P}}_{k,i}^{-1}\geq z_{ t,i}(A,H,R,D,\mathcal{A},\Upsilon, \beta,\delta).
	\end{align}
	where $z_{ t,i}(\cdot )$ is similarly defined as $f_{ t,i}(\cdot)$.
	Let
	\begin{align}\label{eq_ferror}
	\bar f(t,\delta,i)\triangleq &\lambda_{max}\{f_{t,i}(\cdot)-eig_{pos}(z_{t,i}(\cdot))\}-\delta.
	\end{align}
	where $f_{t,i}(\cdot)\triangleq f_{t,i} (A,H,R,D,\mathcal{A},\Upsilon,Q,\bar\beta)$, and $eig_{pos}(z_{t,i}(\cdot))$ generates the diagonal matrix consisting of zeros and positive eigenvalues of $z_{ t,i}(A,H,R,D,\mathcal{A},\Upsilon, \beta,\delta)$, subject to $eig_{pos}(\cdot)\geq z_{ t,i}(\cdot).$
	Then a necessary condition of (\ref{eq_tri}) is 
	$\bar f(t,\delta,i)>0,$ $\forall t\in  [0:t_1-t_0]$.
\section{Proof of Lemma \ref{lem_trigger2}}\label{pf_lemma_tri2}
		If the event in (\ref{estimates_lattest}) is successively not triggered in the interval $[t_2:t_3]$, then
		\begin{align}\label{eq_tri2}
		\lambda_{max}(\tilde P_{k,i}^{-1}-\bar{\tilde{P}}_{k,i}^{-1})\leq \delta , \forall k\in [t_2:t_3].
		\end{align}
		Next, we will find a sufficient condition of (\ref{eq_tri2}). Specifically, we aim to construct an uniformly upper bound of $\tilde P_{k,i}^{-1}$ and an uniformly lower bound of $\bar{\tilde{P}}_{k,i}^{-1}$. Then, we construct a sufficient condition of (\ref{eq_tri2}) based on the bounds.
		
		Similar to the proof method of Lemma \ref{lem_trigger}, for $\tilde P_{k,i}^{-1}$, $k\in [t_2:t_2+t]\subseteq [t_2:t_3]$, there exists a matrix $f_{t,i} (A,H,R,D,\mathcal{A},\Upsilon,Q,\bar\beta),$ such that $\tilde P_{k,i}^{-1}\leq 	f_{t,i} (A,H,R,D,\mathcal{A},\Upsilon,Q,\bar\beta).$ 
		For the prediction matrices $\bar{\tilde{P}}_{k,i}^{-1}$, $  k\in [t_2:t_2+t]\subseteq [t_2:t_3]$. Define an operator $h(\cdot)$ as $h(X)=AXA^T+Q$, then
		\begin{align*}
		\bar{\tilde{P}}_{k,i}^{-1}&=\left(h^{t+1}(\tilde{P}_{t_2-1,i})\right)^{-1}\\
		&\geq \beta^{t+1} (A^{-t-1})^T\tilde{P}_{t_2-1,i}^{-1}A^{-t-1}\\
		&\geq \beta^{t+1} (A^{-t-1})^T(\bar P_{t_2-1,i}^{-1}+H_{i}^TR_{i}^{-1}H_{i})A^{-t-1},\\
		&\geq \beta^{t+1} (A^{-t-1})^TH_{i}^TR_{i}^{-1}H_{i}A^{-t-1},\\
		&\triangleq l_{t}(A,H,R,\beta),
		\end{align*}
		where the last inequality is obtained due to $\bar P_{t_2-1,i}^{-1}\geq 0.$
		Let
		\begin{align}\label{eq_ferror2}
		\bar g(\delta,t,i)\triangleq &\lambda_{max}\{f_{t,i}(\cdot)-l_{t}(A,H,R,\beta)\}-\delta.
		\end{align}
		where $f_{t,i}(\cdot)\triangleq f_{t,i} (A,H,R,D,\mathcal{A},\Upsilon,Q,\bar\beta)$.
		Then a sufficient condition of triggering condition (\ref{eq_tri2}) is 
		$\bar g(\delta,t,i)\leq 0,$ $\forall t\in  [0:t_3-t_2]$.

	%
	%
	\small
%
		\bibliographystyle{IEEEtran}
	\bibliography{IEEEabrv,references_filtering}	

\begin{thebibliography}{10}
\providecommand{\url}[1]{#1}
\csname url@samestyle\endcsname
\providecommand{\newblock}{\relax}
\providecommand{\bibinfo}[2]{#2}
\providecommand{\BIBentrySTDinterwordspacing}{\spaceskip=0pt\relax}
\providecommand{\BIBentryALTinterwordstretchfactor}{4}
\providecommand{\BIBentryALTinterwordspacing}{\spaceskip=\fontdimen2\font plus
\BIBentryALTinterwordstretchfactor\fontdimen3\font minus
  \fontdimen4\font\relax}
\providecommand{\BIBforeignlanguage}[2]{{%
\expandafter\ifx\csname l@#1\endcsname\relax
\typeout{** WARNING: IEEEtran.bst: No hyphenation pattern has been}%
\typeout{** loaded for the language `#1'. Using the pattern for}%
\typeout{** the default language instead.}%
\else
\language=\csname l@#1\endcsname
\fi
#2}}
\providecommand{\BIBdecl}{\relax}
\BIBdecl

\bibitem{cattivelli2010diffusion}
F.~S. Cattivelli and A.~H. Sayed, ``Diffusion strategies for distributed
  \protect{Kalman} filtering and smoothing,'' \emph{IEEE Transactions on
  Automatic Control}, vol.~55, no.~9, pp. 2069--2084, 2010.

\bibitem{Zhang2016Distributed}
L.~Zhang, Z.~Ning, and Z.~Wang, ``Distributed filtering for fuzzy time-delay
  systems with packet dropouts and redundant channels,'' \emph{IEEE
  Transactions on Systems Man \& Cybernetics Systems}, vol.~46, no.~4, pp.
  559--572, 2016.

\bibitem{Hu2012Diffusion}
J.~Hu, L.~Xie, and C.~Zhang, ``Diffusion \protect{Kalman} filtering based on
  covariance intersection,'' \emph{IEEE Transactions on Signal Processing},
  vol.~60, no.~2, pp. 891--902, 2012.

\bibitem{Zhu2016Distributed}
Y.~Zhu, L.~Zhang, and W.~X. Zheng, ``Distributed $\protect{H}_{\infty}$
  filtering for a class of discrete-time markov jump lur¡¯e systems with
  redundant channels,'' \emph{IEEE Transactions on Industrial Electronics},
  vol.~63, no.~3, pp. 1876--1885, 2016.

\bibitem{olfati2009kalman}
R.~Olfati-Saber, ``\protect{Kalman}-consensus filter: Optimality, stability,
  and performance,'' in \emph{Proc. Joint IEEE Conference on Decision and
  Control and Chinese Control Conference}, 2009, pp. 7036--7042.

\bibitem{olfati2007distributed}
------, ``Distributed \protect{Kalman} filtering for sensor networks,'' in
  \emph{Proc. IEEE Conference on Decision and Control}, 2007, pp. 5492--5498.

\bibitem{yang2014stochastic}
W.~Yang, G.~Chen, X.~Wang, and L.~Shi, ``Stochastic sensor activation for
  distributed state estimation over a sensor network,'' \emph{Automatica},
  vol.~50, no.~8, pp. 2070--2076, 2014.

\bibitem{Battistelli2014Kullback}
G.~Battistelli and L.~Chisci, ``Kullback¨c-\protect{Leibler} average, consensus
  on probability densities, and distributed state estimation with guaranteed
  stability,'' \emph{Automatica}, vol.~50, no.~3, pp. 707--718, 2014.

\bibitem{Battistelli2016stability}
------, ``Stability of consensus extended \protect{Kalman} filter for
  distributed state estimation,'' \emph{Automatica}, vol.~68, pp. 169--178,
  2016.

\bibitem{Battistelli2015Consensus}
G.~Battistelli, L.~Chisci, G.~Mugnai, A.~Farina, and A.~Graziano,
  ``Consensus-based linear and nonlinear filtering,'' \emph{IEEE Transactions
  on Automatic Control}, vol.~60, no.~5, pp. 1410--1415, 2015.

\bibitem{Dm2015Distributed}
S.~Das and J.~M.~F. Moura, ``Distributed \protect{Kalman} filtering with
  dynamic observations consensus,'' \emph{IEEE Transactions on Signal
  Processing}, vol.~63, no.~17, pp. 4458--4473, 2015.

\bibitem{yang2017stochastic}
W.~Yang, C.~Yang, H.~Shi, L.~Shi, and G.~Chen, ``Stochastic link activation for
  distributed filtering under sensor power constraint,'' \emph{Automatica},
  vol.~75, pp. 109--118, 2017.

\bibitem{carli2008distributed}
R.~Carli, A.~Chiuso, L.~Schenato, and S.~Zampieri, ``Distributed
  \protect{Kalman} filtering based on consensus strategies,'' \emph{IEEE
  Journal on Selected Areas in Communications}, vol.~26, no.~4, pp. 622--633,
  2008.

\bibitem{ugrinovskii2011distributed}
V.~Ugrinovskii, ``Distributed robust filtering with $\protect{H}_{\infty}$
  consensus of estimates,'' \emph{Automatica}, vol.~47, no.~1, pp. 1--13, 2011.

\bibitem{stankovic2009consensus}
S.~S. Stankovi{\'c}, M.~S. Stankovi{\'c}, and D.~M. Stipanovi{\'c}, ``Consensus
  based overlapping decentralized estimation with missing observations and
  communication faults,'' \emph{Automatica}, vol.~45, no.~6, pp. 1397--1406,
  2009.

\bibitem{zhou2012distributed}
Z.~Zhou, Y.~Hong, and H.~Fang, ``Distributed estimation for moving target under
  switching interconnection network,'' in \emph{Proc. International Conference
  on Control Automation Robotics \& Vision}, 2012, pp. 1818--1823.

\bibitem{zhou2013distributed}
Z.~Zhou, H.~Fang, and Y.~Hong, ``Distributed estimation for moving target based
  on state-consensus strategy,'' \emph{IEEE Transactions on Automatic Control},
  vol.~58, no.~8, pp. 2096--2101, 2013.

\bibitem{kar2011gossip}
S.~Kar and J.~M. Moura, ``Gossip and distributed \protect{Kalman} filtering:
  Weak consensus under weak detectability,'' \emph{IEEE Transactions on Signal
  Processing}, vol.~59, no.~4, pp. 1766--1784, 2011.

\bibitem{Wenling2015Diffusion}
W.~Li, Y.~Jia, J.~Du, and D.~Meng, ``Diffusion \protect{Kalman} filter for
  distributed estimation with intermittent observations,'' in \emph{Proc.
  American Control Conference}, 2015, pp. 4455--4460.

\bibitem{xiao2004fast}
L.~Xiao and S.~Boyd, ``Fast linear iterations for distributed averaging,''
  \emph{Systems \& Control Letters}, vol.~53, no.~1, pp. 65--78, 2004.

\bibitem{wang2017convergence}
S.~Wang and W.~Ren, ``On the convergence conditions of distributed dynamic
  state estimation using sensor networks: A unified framework,'' \emph{IEEE
  Transactions on Control Systems Technology}, 2017.

\bibitem{crassidis2003unscented}
J.~L. Crassidis and F.~L. Markley, ``Unscented filtering for spacecraft
  attitude estimation,'' \emph{Journal of Guidance Control and Dynamics},
  vol.~26, no.~4, pp. 536--542, 2003.

\bibitem{Chandrasekar2004State}
J.~Chandrasekar, O.~Barrero, A.~Ridley, and D.~S. Bernstein, ``State estimation
  for linearized mhd flow,'' in \emph{IEEE Conference on Decision and Control},
  2004, pp. 2584--2589.

\bibitem{wen1992model}
W.~Wen and H.~F. Durrant-Whyte, ``Model-based multi-sensor data fusion,'' in
  \emph{IEEE International Conference on Robotics and Automation}, 1992, pp.
  1720--1726.

\bibitem{shen2006reliable}
S.~Shen, L.~Hong, and S.~Cong, ``Reliable road vehicle collision prediction
  with constrained filtering,'' \emph{Signal Processing}, vol.~86, no.~11, pp.
  3339--3356, 2006.

\bibitem{simon2002kalman}
D.~Simon and T.~L. Chia, ``\protect{Kalman} filtering with state equality
  constraints,'' \emph{IEEE transactions on Aerospace and Electronic Systems},
  vol.~38, no.~1, pp. 128--136, 2002.

\bibitem{rotea2008state}
M.~Rotea and C.~Lana, ``State estimation with probability constraints,''
  \emph{International Journal of Control}, vol.~81, no.~6, pp. 920--930, 2008.

\bibitem{julier2007kalman}
S.~J. Julier and J.~J. LaViola, ``On \protect{Kalman} filtering with nonlinear
  equality constraints,'' \emph{IEEE Transactions on Signal Processing},
  vol.~55, no.~6, pp. 2774--2784, 2007.

\bibitem{Te2009state}
B.~O.~S. Teixeira, J.~Chandrasekar, and L.~A.~B. T{\^o}rres, ``State estimation
  for linear and non-linear equality-constrained systems,'' \emph{International
  Journal of Control}, vol.~82, no.~5, pp. 918--936, 2009.

\bibitem{ko2007state}
S.~Ko and R.~R. Bitmead, ``State estimation for linear systems with state
  equality constraints,'' \emph{Automatica}, vol.~43, no.~8, pp. 1363--1368,
  2007.

\bibitem{rao2001constrained}
C.~V. Rao, J.~B. Rawlings, and J.~H. Lee, ``Constrained linear state estimation
  -- a moving horizon approach,'' \emph{IEEE Transactions on Aerospace and
  Electronic Systems}, vol.~37, no.~10, pp. 1619--1628, 2001.

\bibitem{simon2010kalman}
D.~Simon, ``\protect{Kalman} filtering with state constraints: a survey of
  linear and nonlinear algorithms,'' \emph{IET Control Theory \& Applications},
  vol.~4, no.~8, pp. 1303--1318, 2010.

\bibitem{trimpe2014event}
S.~Trimpe and R.~D'Andrea, ``Event-based state estimation with variance-based
  triggering,'' \emph{IEEE Transactions on Automatic Control}, vol.~59, no.~12,
  pp. 3266--3281, 2014.

\bibitem{wu2013event}
J.~Wu, Q.~Jia, K.~H. Johansson, and L.~Shi, ``Event-based sensor data
  scheduling: Trade-off between communication rate and estimation quality,''
  \emph{IEEE Transactions on Automatic Control}, vol.~58, no.~4, pp.
  1041--1046, 2013.

\bibitem{HanTacStochasEV}
D.~Han, Y.~Mo, J.~Wu, S.~Weerakkody, B.~Sinopoli, and L.~Shi, ``Stochastic
  event-triggered sensor schedule for remote state estimation,'' \emph{IEEE
  Transactions on Automatic Control}, vol.~60, no.~10, pp. 2661--2675, 2015.

\bibitem{shi2015set}
D.~Shi, T.~Chen, and L.~Shi, ``On set-valued \protect{Kalman} filtering and its
  application to event-based state estimation,'' \emph{IEEE Transactions on
  Automatic Control}, vol.~60, no.~5, pp. 1275--1290, 2015.

\bibitem{miskowicz2006send}
M.~Miskowicz, ``Send-on-delta concept: an event-based data reporting
  strategy,'' \emph{Sensors}, vol.~6, no.~1, pp. 49--63, 2006.

\bibitem{liu2015event}
Q.~Liu, Z.~Wang, X.~He, and D.~Zhou, ``Event-based recursive distributed
  filtering over wireless sensor networks,'' \emph{IEEE Transactions on
  Automatic Control}, vol.~60, no.~9, pp. 2470--2475, 2015.

\bibitem{andren2016event}
M.~Andr{\'e}n and A.~Cervin, ``Event-based state estimation using an improved
  stochastic send-on-delta sampling scheme,'' in \emph{Proc. International
  Conference on Event-based Control, Communication, and Signal Processing},
  2016, pp. 1--8.

\bibitem{nguyen2007improving}
V.~Nguyen and Y.~Suh, ``Improving estimation performance in networked control
  systems applying the send-on-delta transmission method,'' \emph{Sensors},
  vol.~7, no.~10, pp. 2128--2138, 2007.

\bibitem{Cat2010Diffusion}
F.~S. Cattivelli and A.~H. Sayed, ``Diffusion strategies for distributed
  \protect{Kalman} filtering and smoothing,'' \emph{IEEE Transactions on
  Automatic Control}, vol.~55, no.~9, pp. 2069--2084, 2010.

\bibitem{liu2015eventinformatics}
Q.~Liu, Z.~Wang, X.~He, and D.~Zhou, ``Event-based distributed filtering with
  stochastic measurement fading,'' \emph{IEEE Transactions on Industrial
  Informatics}, vol.~11, no.~6, pp. 1643--1652, 2015.

\bibitem{Battistelli2016Distributed}
G.~Battistelli, L.~Chisci, and D.~Selvi, ``Distributed \protect{Kalman}
  filtering with data-driven communication,'' in \emph{Proc. International
  Conference on Information Fusion}, 2016, pp. 1042--1048.

\bibitem{Julier1997A}
S.~J. Julier and J.~K. Uhlmann, ``A non-divergent estimation algorithm in the
  presence of unknown correlations,'' vol.~4, 1997, pp. 2369--2373.

\bibitem{horn2012matrix}
R.~A. Horn and C.~R. Johnson, \emph{Matrix analysis}.\hskip 1em plus 0.5em
  minus 0.4em\relax Cambridge university press, 2012.

\bibitem{varga2009matrix}
R.~S. Varga, \emph{Matrix iterative analysis}.\hskip 1em plus 0.5em minus
  0.4em\relax Springer Science \& Business Media, 2009.

\bibitem{lou2013convergence}
Y.~Lou, G.~Shi, K.~H. Johansson, and Y.~Hong, ``Convergence of random sleep
  algorithms for optimal consensus,'' \emph{Systems \& Control Letters},
  vol.~62, no.~12, pp. 1196--1202, 2013.

\end{thebibliography}

\end{document}